%% file: IDM.tex
\documentclass[onefignum,onetabnum]{siamonline190516}

\ifpdf
\hypersetup{
  pdftitle={Limitations and Improvements of the Intelligent Driver Model (IDM)},
  pdfauthor={Saleh Albeaik, Alexandre Bayen, Maria Teresa Chiri, Xiaoqian Gong, Amaury Hayat, Nicolas Kardous, Alexander Keimer, Sean T. McQuade, Benedetto Piccoli, Yiling You}
}
\fi

\usepackage[utf8]{inputenc}
\usepackage{amsmath,amssymb,amsfonts}
\usepackage{mathtools}
\usepackage{dsfont}
\usepackage{pgfplots}
\usepackage{enumitem}
\usepackage{ulem}
\pgfplotsset{compat=newest}
\pgfplotsset{plot coordinates/math parser=false}

\usepackage[capitalize,nameinlink]{cleveref}
\usepackage{todonotes}

\graphicspath{ {images/} }
\allowdisplaybreaks

\input{definitions.tex}

\input{define_cars.tex}

\newsiamremark{remark}{Remark}
\newsiamremark{example}{Example}
\newsiamremark{assumption}{Assumption}

\makeatletter
\def\mathcolor#1#{\@mathcolor{#1}}
\def\@mathcolor#1#2#3{%
  \protect\leavevmode
  \begingroup
    \color#1{#2}#3%
  \endgroup
}
\makeatother

\title{Limitations and Improvements of the Intelligent Driver Model (IDM)}
\author{Saleh Albeaik\thanks{Department of Civil and Environmental Engineering, University of California, Berkeley, 94720 Berkeley, CA, USA}, \email{albeaik@berkeley.edu} \and Alexandre Bayen\thanks{Institute for Transportation Studies (ITS), University of California, Berkeley, 94720 Berkeley, CA, USA}, \email{bayen@berkeley.edu} \and Maria Teresa Chiri\thanks{Department of Mathematics, Penn State University, 16802 University Park, PA, USA}, \email{mxc6028@psu.edu} \and Xiaoqian Gong\thanks{School of Mathematical and Statistical Science, Arizona State University, 85281 Tempe, AZ, USA}, \email{xiaoqian.gong@asu.edu} \and Amaury Hayat\thanks{Centre d'Enseignement et de Recherche en Mathématiques et Calcul Scientifique (CERMICS), Ecole des Ponts, 77455 Marne-la-Vallée, France, \email{amaury.hayat@enpc.fr}
} \and Nicolas Kardous\thanks{Department of Industrial Engineering \& Operations Research, University of California Berkeley, Berkeley, CA, USA}, \email{nicolas.kardous@berkeley.edu} \and Alexander Keimer \and Sean T.~McQuade\thanks{Department of Mathematical Sciences and Center for Computational and Integrative Biology, Rutgers University, Camden, NJ, USA}, \email{sean.mcquade@rutgers.edu} \and Benedetto Piccoli\thanks{Department of Mathematical Sciences and Center for Computational and Integrative Biology, Rutgers University, Camden, NJ, USA}, \email{piccoli@camden.rutgers.edu} \and Yiling You\thanks{Department of Mathematics, University of California, Berkeley, 94720 Berkeley, CA, USA}, \email{yiling.you@berkeley.edu}
}

\begin{document}
\maketitle

\begin{abstract}
    This contribution analyzes the widely used and well-known ``intelligent driver model'' (briefly IDM), which is a second order car-following model governed by a system of ordinary differential equations. Although this model was intensively studied in recent years for properly capturing traffic phenomena and driver braking behavior, a rigorous study of the well-posedness has, to our knowledge, never been performed. First it is shown that, for a specific class of initial data, the  vehicles' velocities become negative or even diverge to \(-\infty\) in finite time, both undesirable properties for a car-following model. Various modifications of the IDM are then proposed in order to avoid such ill-posedness.  The theoretical remediation of the model, rather than \textit{post facto} by ad-hoc modification of code implementations, allows a more sound numerical implementation and preservation of the model features.
    Indeed, to avoid inconsistencies and ensure dynamics close to the one of the original model, one may need to inspect and clean large input data, which may result in practically impossible scenarios for large-scale simulations. Although well-posedness issues might only occur for specific initial data, this may happen frequently when different traffic scenarios are analyzed, and especially in presence of lane-changing, on ramps and other network components as it is the case for most commonly used micro-simulators.
    On the other side, it is shown that well-posedness can be guaranteed by straight-forward improvements, such as those obtained by slightly changing the acceleration to prevent the velocity from becoming negative.
\end{abstract}
\begin{keywords}
  IDM, Intelligent driver model, system of ODEs, discontinuous ODEs, traffic modelling, microscopic traffic modelling, car following model, well-posedness of ODEs, existence and uniqueness of solutions of ODE;
\end{keywords}

\begin{AMS}
 34A12, 34A38, 65L05, 65L08
\end{AMS}
\section{Introduction}
The field of car following modeling historically goes back to the early 1950s' (and probably before) \cite{chandler,gazis1961nonlinear}. Most of the early work in this field focused on establishing the model equations, without paying much attention to the mathematical framework required to characterize solutions to the resulting ordinary differential equations (ODEs) describing the motion of the vehicles. 
The models are mainly classified into \textit{acceleration models} for longitudinal movement, \textit{lane-changing models} for lateral movement and \textit{decisional models} for discrete-choice situations.
Among all the car-following models introduced so far (\cite{bando1995phenomenological,bexelius1968extended,brackstone1999car, orosz2005bifurcations,orosz2006subcritical} to just name a few) it is worth mentioning the Gazis-Herman-Rothery (GHR) model \cite{chandler} which determines the relative velocity between two-lane based vehicles, the Safe Distance Model \cite{GIPPS}, the Optimal Velocity Model \cite{Bando} in which the acceleration of the single vehicle is controlled according to the velocity of the leading vehicle, and the Intelligent Driver Model (IDM) which is subject of analysis in the present work.
For a comprehensive overview of the main car-following models we refer to \cite{Treiber2012,Matcha2020,hoogendoorn2001state}.

The IDM has been introduced in \cite{treiber2000congested} and is a deterministic time-continuous model describing  the dynamics of the positions and velocities of every vehicle.  Similarly to any car-following model, the idea behind it is that drivers control their vehicles to react to the stimulus from preceding vehicles.
It aims to balance two different aspects, the necessity to keep safe separation with the vehicle in front and the desire to achieve ``free flow'' speed. This model presents some peculiarities which made it subject of intense research in the last two decades. Indeed, it is constructed to be \textit{collision-free}, all the parameters can be interpreted and empirically measured, the stability of the model can be calibrated to empirical data, and there exists an equivalent macroscopic counterpart \cite{HELBING2002}.
In literature we find many extensions of the original IDM, each of which seeks to incorporate new realistic features.
	The Enhanced IDM \cite{Kesting_2010} presents an improved heuristic of the IDM useful for multi-lane simulations, which prevents the model from ``over-reactions" even when the driver of the leading vehicle suddenly brakes with the maximum possible deceleration. 
	
	The Foresighted Driver Model (FDM) starts from the IDM and assumes that a driver acts in a way that balances predictive risk (due to possible collisions along his route) with utility (time required to travel, smoothness of the ride) \cite{Eggert}.
	
	Other extensions of the IDM aim to improve the driver safety and to respect the vehicle capability \cite{derbel}, to strengthen the power of each vehicle in proportion to the immediately preceding vehicle \cite{Zhipeng}, and to incorporate the spatially varying velocity profile to account the variation in different types of maneuvers through intersection \cite{Liebner}.
	Another natural extension is given by  Multi-anticipative IDM \cite{TREIBER2006} which models the reaction of a driver to several vehicles ahead just by summing up the corresponding vehicle-vehicle pair interactions with the same weight coefficients. 
	
	More recently,  stochastic versions of the IDM have been introduced: to describe a probabilistic motion prediction applicable for long term trajectory planning \cite{Hoermann}, to study mechanisms behind traffic flow instabilities, indifference regions of finite human perception thresholds and external noise \cite{TREIBER2017}, and to incorporate context-dependent upper and lower bounds on acceleration \cite{Schulz1,Schulz2}.

Throughout the decades, most of the engineering community worked on improving the ability of the models to capture specific behavioral phenomena, at the expense of the characterization of the solutions. Thus, to this day, only a few articles use models, and corresponding solutions, that are well characterized in terms of existence, uniqueness, and regularity. An example of this practice is provided by the double integrator $\ddot{x}(t)=u\in U$, where $U$ is the input set, used abundantly as a canonical example in numerous control articles. 
On the other hand, when models have inherent flaws leading to unbounded or undefined solutions, as in the case of unbounded acceleration, ad hoc methods have been traditionally applied \textit{post facto} by engineering the numerical implementations. For instance, in commonly used microsimulation tools, such as SUMO \cite{SUMO2018}, Aimsun \cite{aimsun} and others, unbounded quantities are clipped, leading to ``acceptable'' numerical solutions. However, in the process, the fidelity to the original model is compromised, and the numerical simulations may not represent any instantiation of the model. Consequently, the properties of the considered continuous model might be lost as well. Finally, the process of clipping can introduce additional issues, not necessarily present in the original model, and prevent the definition of any theoretical models corresponding to the obtained numerical simulations. The present article thus attempts to provide a full pipeline in which the model is first mathematically well defined (including existence, uniqueness and regularity characterization of the solutions), and then numerically implemented using appropriate numerical differentiation schemes. The final achievement is a thorough correspondence between theory and implementation. The well-posedness of car-following models is also fundamental to model the traffic flow from the mean-field perspective, see \cite{fornasier2014mean,gong2020mean}.

\subsection{The aim of this contribution}
The introduced IDM has two mathematical and modelling drawbacks:
\begin{itemize}
    \item The velocities of specific vehicles might become negative at specific times, which might not be desirable from a modelling point of view. 
    \item The velocities of specific vehicles might diverge to \(-\infty\) in finite time, so that the solution of the system of ODE's ceases to exist. 
\end{itemize}
We will discuss these drawbacks and determine under which conditions on the initial datum and parameters they can happen. We will present several improvements so that the solutions exist on every finite time horizon.
 \subsection{Structure of this article}
 	The paper is organized in the following way. In \Cref{sec:IDM} we review the classical Intelligent Driver Model (IDM) and describe briefly the physical meaning of the parameters involved. Well-posedness of the IDM for small time horizon is stated in \cref{theo:well_posedness_small_time_horizon}. In \Cref{sec:counterexamples} we analyze peculiar and possibly pathological behaviors of the model. Specifically, we provide explicit settings in which it produces negative velocities (\cref{ex:negative_velocity}), negative velocities and blow-up of the solution in finite time (\cref{ex:negative_velocity_blowup}, \cref{ex:negative_velocity_blow_down}).
 	\Cref{sec:well_pos_spec_dat_times} collects some of the main results of this work. Existence and uniqueness of a solution for small times with lower bound on the distance which can be interpreted as ``collision free.''
 	\Cref{sec:IDM_improvements} is devoted to the exploration, analysis and comparison of adjustments to the classic IDM in order to avoid the problems mentioned in \Cref{sec:counterexamples} for general initial data. To this end, we introduce modified versions of the IDM for which well-posedness is proved: the \textit{projected IDM}, the \textit{acceleration projected IDM}, and the \textit{velocity regularized acceleration IDM} defined respectively in  \cref{defi:IDM_projected_velocities}, \cref{defi:IDM_projected_velocities_acceleration}, \cref{defi:IDM_velocity_regularized}.
 	 A further and more drastic adjustment to the classic model is proposed in \cref{defi:IDM_discontinuous} which involves a discontinuous acceleration, and therefore is denoted as \textit{discontinuous improvement}. 
 	 In \Cref{sec:multi_vehicle} we present some well-posedness results for the many vehicle case (based on the previous analysis in \Cref{sec:IDM_improvements}, and
 	 finally, in \Cref{sec:conclusions} we draw conclusions from our work and mention possible research directions opened by this contribution.

 \section{The intelligent driver model (IDM): Definitions and basic results}\label{sec:IDM}
 In this section we introduce the intelligent driver model (IDM) as the following system of ordinary differential equations. To this end, we require to define the acceleration function as follows:
\begin{definition}[The IDM acceleration]\label{defi:IDM_acc}
Let $T\in\R_{>0}$ be fixed. For a parameter set \((a,b,v_{\text{free}},\tau,s_{0},l,\delta)\in\R^{3}_{>0}\times(0,T)\times\R_{>0}^{2}\times \R_{>1}\) we define the following IDM car-following acceleration on the set
\begin{align*}
    \mathcal{A}&\:\{(x,v,\xl,\vl)\in\R^{4}:\  \xl-x-l>0\},\\
    \Acc&:\begin{cases}\mathcal{A}&\rightarrow\R\\
    (x,v,\xl,\vl)&\mapsto a\bigg(1-\big(\tfrac{|v|}{v_{\text{free}}}\big)^{\delta}-\Big(\tfrac{2\sqrt{ab}(s_{0}+v\tau)+v(v-\vl)}{2\sqrt{ab}(\xl-x-l)}\Big)^{2}\bigg).
    \end{cases}
\end{align*}
\end{definition}
\begin{remark}[Absolute values in the IDM Acceleration]
It is worth mentioning that in most literature the parameter \(\delta\) is not precisely specified except that it is assumed to be positive. However, as we will show, velocities can become negative, and this is why we assumed that the acceleration term in \cref{defi:IDM_acc} involves the absolute value of the velocity so that
\begin{equation}
\big(\tfrac{|v|}{v_{\text{free}}}\big)^{\delta}\label{eq:42}
\end{equation}
is well defined for all \(v\in\R\) and \(\delta\in\R_{>0}\). Obviously, this is only one choice, and it might be more reasonable to replace it by
\begin{equation}
\sgn(v)\big(\tfrac{|v|}{v_{\text{free}}}\big)^{\delta}\label{eq:42_1}
\end{equation}
so that this term contributes positive to the acceleration for negative velocities and indeed counteracts a negative velocity. For \(\delta\in 2\N_{\geq1}+1\) \cref{eq:42_1} can actually be replaced by the version without the absolute value and the same is true for \(\delta\in 2\N_{\geq1}\) with the drawback that this part of the acceleration will always remain negative as we also assume for now in \cref{defi:IDM_acc}.

An analysis similar to the one in this paper can then be carried out with mentioning that although in this case the solution's velocity can diverge to \(-\infty\).
\end{remark}
As we will require for the leader a specific acceleration, the ``free-flow acceleration,'' we define as follows
\begin{definition}[Free flow acceleration]\label{defi:free_flow}
For \((a,v_{\text{free}},\delta)\in\R_{>0}^{2}\times\R_{>1}\) the \textbf{free flow} acceleration is defined by
\begin{align*}
    \Accfront:&\begin{cases}\R^{2}&\rightarrow\R\\
   (x,v)&\mapsto a\Big(1-\big(\tfrac{|v|}{v_{\text{free}}}\big)^{\delta}\Big).
   \end{cases}
   \end{align*}
\end{definition}
Having defined the acceleration function, we are ready to present \cref{defi:IDM_model}:
\begin{definition}[The IDM]\label{defi:IDM_model}
Given \cref{defi:IDM_acc}, we call the following system of ordinary differential equations in position \(\bx=(\xl, x):[0,T]\rightarrow\R^{2}\) and velocity \(\bv=(\vl, v):[0,T]\rightarrow\R^{2}\)
\begin{equation}
\begin{aligned}
    \dot{\xl}(t) &=\vl(t), && t \in [0,T],\\
    \dot{\vl}(t)&=\ul(t),&& t\in[0,T],\\
    \dot{x}(t) &=v(t), && t \in [0,T],\\
   \dot{v}(t)&=\Acc(x(t),v(t),\xl(t),\vl(t)), && t\in[0,T], \\
    (\xl(0),x(0)) &=(\xlz,x_{0}),\\
    (\vl(0),v(0)) &= (\vlz,v_{0}),
\end{aligned}
\end{equation}
with leading dynamics \(\ul:[0,T]\rightarrow\R\) the car-following \textbf{IDM}. \((x_0, \xlz, v_0, \vlz)\in\R^{2}\times\R_{\geq 0}^{2}\) are initial positions and velocities. 
\end{definition}
This is schematically illustrated in \cref{fig:car_following}.
\begin{figure}
    \centering
\begin{tikzpicture}[scale=0.6]
 \draw[fill,color=gray!20!white](-.5,-1)--(-0.5,2.25)--(7,2.25)--(7,-1)--cycle;
 \draw[->] (-0.5,-1) -- (-0.5,2.5);
\draw (-0.75,1) node {\(t\)};
\draw[->] (-0.5,-1) -- (7.25,-1);
\draw[dashed,blue] (1.25,-0.5) -- (4.5,1.75);
\draw [dashed,red] plot [smooth] coordinates { (4,-0.5) (5.25,0.5) (6.25,1.5)};
\begin{scope}[shift={(1.5,0)}]\car{blue}\end{scope}
  \begin{scope}[shift={(5,1)}]\car{red}\end{scope}
 \begin{scope}[shift={(3,1)}]\car{blue}\end{scope}
 \begin{scope}[shift={(4,0)}]\car{red}\end{scope}
 \begin{scope}[shift={(0,-1)}]\car{blue}\end{scope}
 \begin{scope}[shift={(3,-1)}]\car{red}\end{scope}
\end{tikzpicture}
\caption{For \(t\in[0,T]\) the leader \textcolor{red}{\(\xl(t)\)} with its dynamics determined by the acceleration \(u_{\text{lead}}(t)\) and the follower \textcolor{blue}{\(x(t)\)} with its dynamics governed by the classical IDM (acceleration \(\Acc\) as in \cref{defi:IDM_acc}). The overall dynamics is stated in \cref{defi:IDM_model}. The \textcolor{blue}{follower} approaches the \textcolor{red}{leader}. Will they collide?}
    \label{fig:car_following}
\end{figure}
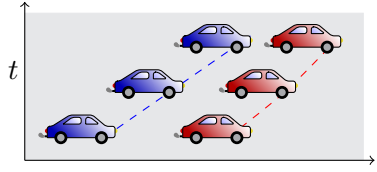
We give a short overview of some of the meanings of the parameters in the IDM.
\begin{remark}[Meaning of the previously introduced parameters] The parameters $a$, $b$, $v_{\text{free}}$, $\tau$, $s_0$, $l$ and $\delta$, introduced in \cref{defi:IDM_acc}, are model parameters which have -- according to  \cite{treiber2000congested} -- the following meaning:
\begin{description}
   \item[acceleration] $a$: the maximum vehicle acceleration;
   \item[comfortable braking deceleration] $b$: a positive number; 
   \item[desired velocity] $v_{\text{free}}$: the velocity the vehicle would drive at in free traffic;
   \item[desired time headway] $\tau$: the minimum possible time to the vehicle in front;
   \item[minimum spacing] $s_0$: a minimum desired net distance;
    \item[the length of the vehicle] $l$;
    \item[the acceleration exponent] $\delta$: Specifying how the acceleration decreases when approaching the desired velocity $v_{\text{free}}$.
\end{description}
\cref{tab:parameters} shows some suggested values for the parameters already identified in \cite{treiber2000congested}.
\begin{table}
\centering
	\caption{According to \cite{treiber2000congested} typical and physical meaningful variables for the IDM}
    	\begin{tabular}{l||r||r }
		\textbf{Parameters} &\textbf{Variable} & \textbf{Suggested value}\\
		\hline\hline
		Maximum acceleration& $a$   & 0.73 $m/s^2$   \\
		Desired deceleration& $b$   &   1.67 $m/s^2$  \\
		Desired velocity& $v_{\text{free}}$ & 120 $km/h$\\
		Desired time headway& $\tau$   & 1.6 $s$ \\
		Minimum spacing & $s_0$ & 2 $m$   \\
		Length of the vehicle& $l$ & 5 $m$ \\
		Acceleration exponent & $\delta$&  4 \\
    \end{tabular}
	\label{tab:parameters}
\end{table}
\end{remark}

\smallskip
For the system to be physically reasonable we require some additional assumptions on the order of the initial position and other parameters for the acceleration functions. This is made precise in the following \cref{ass:input_datum}.
\begin{assumption}[Assumptions on input datum and more]\label{ass:input_datum}
We assume that 
\begin{description}
\item[Leading velocity:]
\(
  \ul\in \mathcal{U}_{\text{lead}}\:\Big\{u\in L^{\infty}((0,T)): \vlz+\int_{0}^{t}u(s)\dd s\geq0\ \forall t\in[0,T]\Big\}.
\)
\item[Input parameters for \(\Acc\):] \((a,b,v_{\text{free}},\tau,s_{0},l,\delta)\in\R^{3}_{>0}\times(0,T)\times\R_{>0}^{2}\times\R_{>1}.\)
\item[Physical relevant initial datum:] \((x_{0},\xlz,v_{0},\vlz)\in\R^{2}\times\R_{\geq 0}^{2}:\ x_{0}<\xlz-l\).
\end{description}
\end{assumption}
The previous assumption on the involved datum enables it to prove the well-posedness on a significantly small time horizon, i.e., that there exists a solution on the time horizon and that this solution is unique:
\begin{theorem}[Well-posedness of sufficiently small time horizon]\label{theo:well_posedness_small_time_horizon}
Given \(N\in\N_{\geq0}\) and \cref{ass:input_datum}, there exists a small enough time \(T^{*}\in\R_{>0}\) so that the IDM in \cref{defi:IDM_model} admits a unique solution \((x,\xl)\in W^{2,\infty}([0,T^{*}])^{2}\).
\end{theorem}
\begin{proof}
The right hand side of \cref{defi:IDM_model} is around the initial datum in \cref{ass:input_datum} locally Lipschitz-continuous. The existence and uniqueness on a small time horizon then follows by the Picard-Lindel\"of Theorem (\cite[Chapter 4]{Arnold} or \cite[Thm 1.3]{Coddington}).
\end{proof}

\section{Counterexamples}\label{sec:counterexamples}
Given \cref{theo:well_posedness_small_time_horizon}, the next natural questions consist of whether the solution on the small time horizon can be extended to any finite time horizon and whether the model remains reasonable. As it turns out, neither points hold if we do not restrict our initial datum beyond \cref{ass:input_datum}.
Of course, such unreasonable behavior appears with specific initial datum, however this may in fact happen frequently when different traffic scenarios are analyzed. Especially, such initial datum could happen in the presence of lane-changing, on ramps and other network components as it is the case for most commonly used micro-simulators.
We present the shortcomings in the following.
\subsection{Negative velocity}
In this subsection, we show that the IDM can develop negative velocities for the following vehicle, although the leading vehicle might drive with positive speed. The reason for this is that if the following vehicle is too close to the leading vehicle, it needs to slow down. Assume now that it actually has already zero velocity, it will need to move backwards to make it to the ``safety'' distance \(s_{0}\) it aims for.
\begin{example}[Negative velocity]\label{ex:negative_velocity}
Assume that 
\(x_0=\xlz-l-\eps\) for \(\eps\in\R_{>0}\) yet to be determined and \(v_0=0\). Then, we compute the change of velocity for the following vehicle and have for \(t\in[0,T]\) according to \cref{defi:IDM_model}
\begin{align*}
    \dot{v}(t)=a\bigg(1- \big(\tfrac{|v(t)|}{v_{\text{free}}}\big)^{\delta}-\Big(\tfrac{2\sqrt{ab}(s_{0}+v(t)\tau)+v(t)(v(t)-\vl(t))}{2\sqrt{ab}(\xl(t)-x(t)-l)}\Big)^{2}\bigg).
\end{align*}
Plugging in \(t=0\) leads to
\begin{align*}
    \dot{v}(0)=a\bigg(1 -\Big(\tfrac{2\sqrt{ab}s_{0}}{2\sqrt{ab}\eps}\Big)^{2}\bigg)=a\Big(1-\big(\tfrac{s_{0}}{\eps}\big)^{2}\Big).
\end{align*}
Thus, whenever \(\epsilon<s_{0}\), the following vehicle has -- at least for small time horizon --  a negative speed, although the leading vehicle drives with arbitrary speed \(\vl\in\R\).
This is also detailed in the following \cref{fig:example_negative_speed_1} and in \cref{fig:example_negative_speed_2} it is demonstrated that for larger spacing this does not occur.
\begin{figure}
\centering
          \begin{tikzpicture}[scale=0.6]
          \begin{axis}[
            width=0.5\textwidth,
            xlabel={\(t\)},
            ylabel={\((\mathcolor{blue}{x(t)},\mathcolor{red}{\xl(t)})\)},
                 legend pos={north west},
                 ymin=-1,ymax=9]
          \addplot[color=red,line width=2pt, mark=none]  table [y index=1, x index=0, col sep=comma] {Tikz_data/Example_3_1_11.csv};
          \addplot[color=blue,line width=2pt, mark=none]  table [y index=2, x index=0, col sep=comma] {Tikz_data/Example_3_1_11.csv};
          \addlegendentry{leader};
          \addlegendentry{follower};
          \end{axis}
          \end{tikzpicture}
          \begin{tikzpicture}[scale=0.6]
          \begin{axis}[
            width=0.5\textwidth,
            xlabel={\(t\)},
            ylabel={\((\mathcolor{blue}{v(t)},\mathcolor{red}{\vl(t)})\)},
                 legend pos={north west},
                 ymin=-0.3,ymax=1.1]
          \addplot[color=red,line width=2pt, mark=none]  table [y index=1, x index=0, col sep=comma] {Tikz_data/Example_3_1_21.csv};
          \addplot[color=blue,line width=2pt, mark=none]  table [y index=2, x index=0, col sep=comma] {Tikz_data/Example_3_1_21.csv};
          \addplot[color=black, domain=0:3, line width=2pt,dotted]{0};
          \end{axis}
          \end{tikzpicture}
               \begin{tikzpicture}[scale=0.6]
          \begin{axis}[
            width=0.5\textwidth,
            xlabel={\(t\)},
            ylabel={\(\mathcolor{blue}{a(t)}\)},
                 legend pos={north west},
                 ymin=-1,ymax=0.5]
          \addplot[color=blue,line width=2pt, mark=none]  table [y index=1, x index=0, col sep=comma] {Tikz_data/Example_3_1_31.csv};
          \addplot[color=black, domain=0:3, line width=2pt,dotted]{0};
          \end{axis}
          \end{tikzpicture}
              \caption{The IDM with parameters \(a=1,\ b=2,\ v_{\text{free}}=1,\ \tau=1.6,\ l=4,\ s_{0}=2,\ d=4\) and datum \(x_{0}=0,\ \xlz=l+1.5<l+s_{0},\ v_0=0,\ \vlz=0\). \textbf{Left:} vehicles' positions, \textbf{middle: } vehicles' velocities and \textbf{right: } followers acceleration. The \textcolor{red}{leader} follows the free flow acceleration profile as in \cref{defi:free_flow}.  As the initial distance between the two vehicles is smaller than \(s_{0}\), the \textcolor{blue}{following vehicle} moves backwards to increase the space.}
              \label{fig:example_negative_speed_1}
\end{figure}
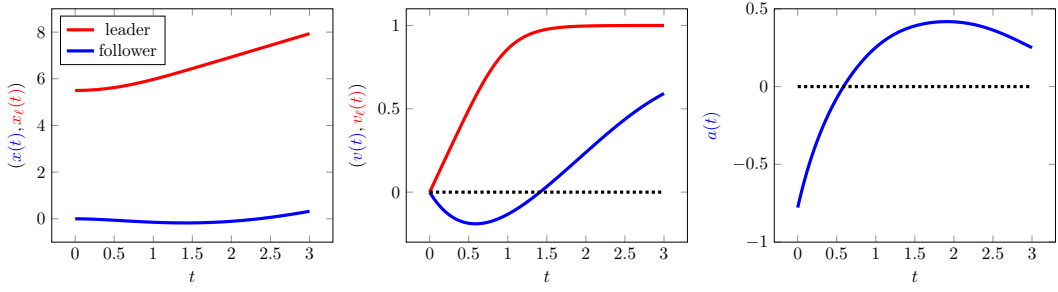

\begin{figure}
\centering
          \begin{tikzpicture}[scale=0.6]
          \begin{axis}[
            width=0.5\textwidth,
            xlabel={\(t\)},
            ylabel={\((\mathcolor{blue}{x(t)},\mathcolor{red}{\xl(t)})\)},
                 legend pos={north west},
                 ymin=-1,ymax=9]
          \addplot[color=red,line width=2pt, mark=none]  table [y index=1, x index=0, col sep=comma] {Tikz_data/Example_3_1_12.csv};
          \addplot[color=blue,line width=2pt, mark=none]  table [y index=2, x index=0, col sep=comma] {Tikz_data/Example_3_1_12.csv};
          \addlegendentry{leader};
          \addlegendentry{follower};
          \end{axis}
          \end{tikzpicture}
          \begin{tikzpicture}[scale=0.6]
          \begin{axis}[
            width=0.5\textwidth,
            xlabel={\(t\)},
            ylabel={\((\mathcolor{blue}{v(t)},\mathcolor{red}{\vl(t)})\)},
                 legend pos={north west} ,
                 ymin=-0.3,ymax=1.1]
          \addplot[color=red,line width=2pt, mark=none]  table [y index=1, x index=0, col sep=comma] {Tikz_data/Example_3_1_22.csv};
          \addplot[color=blue,line width=2pt, mark=none]  table [y index=2, x index=0, col sep=comma] {Tikz_data/Example_3_1_22.csv};
          \addplot[color=black, domain=0:3,, line width=2pt,dotted]{0};
          \end{axis}
          \end{tikzpicture}
                         \begin{tikzpicture}[scale=0.6]
          \begin{axis}[
            width=0.5\textwidth,
            xlabel={\(t\)},
            ylabel={\(\mathcolor{blue}{a(t)}\)},
                 legend pos={north west},
                 ymin=-1,ymax=0.5]
          \addplot[color=blue,line width=2pt, mark=none]  table [y index=1, x index=0, col sep=comma] {Tikz_data/Example_3_1_32.csv};
                  \addplot[color=black, domain=0:3,, line width=2pt,dotted]{0};
          \end{axis}
          \end{tikzpicture}
        \caption{Continuation of \cref{fig:example_negative_speed_1}: The IDM with parameters \(a=1,\ b=2,\ v_{\text{free}}=1,\ \tau=1.6,\ l=4,\ s_{0}=2,\ d=4\) and datum \(x_{0}=0,\ \xlz=l+2=l+s_{0},\ v_0=0,\ \vlz=0 \). The \textcolor{red}{leader} follows the free flow acceleration as in \cref{defi:free_flow}. The initial distance between the two vehicles is large enough so that the \textcolor{blue}{following vehicle} does not attain negative velocity.}
    \label{fig:example_negative_speed_2}
\end{figure}
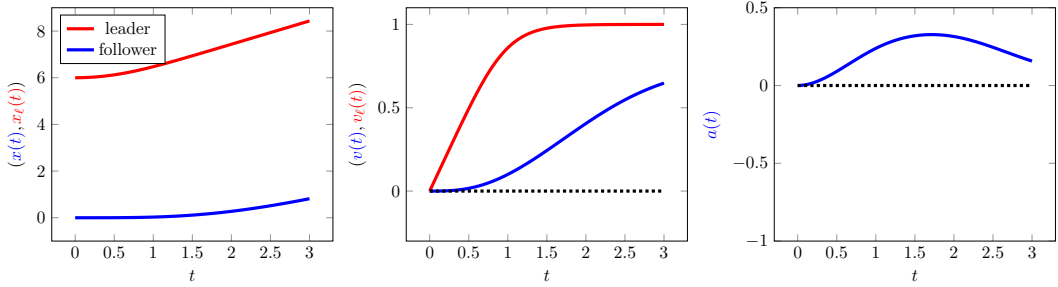

\end{example}

A more reasonable approach for avoiding this type of behavior is that the following car just waits until the leading car has moved farther away. This can be achieved by adjusting the model accordingly as done in \Cref{sec:IDM_improvements}.
\subsection{Velocity exploding in finite time}
In this subsection, we show that the solution can cease to exist in finite time. We first present an example, with fixed parameters, to explain the reasons behind this phenomenon. Then, we generalize the example and illustrate how this phenomenon may occur for parameters in a whole region of the space.
\begin{example}[Negative velocity and a blow-up of the solution in finite time]\label{ex:negative_velocity_blowup}
Assume the constants and initial data are as in \cref{ex:negative_velocity}, with
\(0<\epsilon<1\) and non-negative initial velocity of the leading vehicle $\vlz \geq 0$.
The leading vehicle's position is given by
\begin{align*}
    \xl(t)= \xlz+\vlz t+\tfrac{u}{2}t^{2},
\end{align*}
and the leading vehicle's velocity by
\begin{align*}
    \vl(t) = \vlz + u t.
\end{align*}
Plugging this into the equations for the follower, we obtain the following system of ODEs
\begin{align*}
    \dot{x}(t)&=v(t)\\
    \dot{v}(t)&=a\bigg(1- \Big(\tfrac{|v(t)|}{v_{\text{free}}}\Big)^{\delta}-\Big(\tfrac{2\sqrt{ab}(s_{0}+v(t)\tau)+v(t)(v(t)-\vl(t))}{2\sqrt{ab}\cdot( \xl(t)-\xlz+\epsilon-\int_{0}^{t}v(s)\dd s)}\Big)^{2}\bigg)\\
    x(0)&= \xlz-l-\epsilon\\
    v(0)&=0.    
\end{align*}
Now we fix \(v_{\text{free}}=1=a\), \(s_{0}=16\), \(b=\tfrac{1}{4a}\), \(\tau=8\), \(\vlz=0\), \(u=0\), \(\delta=4\) and \(\xlz=0\), which gives 
\begin{align*}
    \dot{x}(t)&=v(t)\\
    \dot{v}(t)&=1- \big(v(t)\big)^{4}-\Big(\tfrac{(4+v(t))^{2}}{\epsilon-\int_{0}^{t}v(s)\dd s}\Big)^{2}.
\end{align*}
We now show that there exists $\bar{t}$ such that $v(\bar{t})< -1$.
Assume, by contradiction, that $v(t)\geq -1$ on $t\in[0,1]$, 
then since $\epsilon<1$ for $t\in[0,1]$ we get
\begin{equation}
\dot{v}(t) \leq 
1-\Big(\tfrac{(4+v(t))^{2}}{\epsilon-\int_{0}^{t}v(s)\dd s}\Big)^{2}\leq 1-\big(\tfrac{4}{\epsilon+1}\big)^2\leq
1 - \big(2\big)^2 = - 3,
\end{equation}
and by integration and by the previous assumption we have
\[
-1\leq v_{0}+\int_{0}^{t}\dot{v}(s)\dd s\leq v_{0}-3t
\]
and therefore -- plugging in for instance \(t=1\)
\[
-1\leq \int_{0}^{1}\dot{v}(s)\dd s\leq -3,
\]
a contradiction.
%

\paragraph{Blow-up of the solution in finite time}\label{para:blow_up}
Let \(t^{*}\in(0,T]\) denote the first time such that \(v(t^{*})<-1\).
Then, recalling the semi-group property of ODEs, we can consider the initial value problem in \(\tilde{v}:[t^{*},T]\rightarrow\R\)
\begin{align*}
    \tilde{v}(t^{*})&=v(t^{*})<-1\\
    \dot{\tilde{v}}(t)    &=1- \big(\tilde{v}(t)\big)^{4}-\Big(\tfrac{(4+\tilde{v}(t))^{2}}{\epsilon-\int_{0}^{t}\tilde{v}(s)\dd s}\Big)^{2} && t\in[t^{*},T].
\end{align*}
Estimating \(\dot{\tilde{v}},\) we have
\begin{align*}
    \dot{\tilde{v}}(t)\leq 1-\big(\tilde{v}(t)\big)^{4}\quad \forall t\in[t^{*},T].
\end{align*}
As the initial value \(\tilde{v}(t^{*})<-1\) we obtain that \(\tilde{v}\) is monotonically decreasing 
and we can thus estimate
\begin{align*}
    \dot{\tilde{v}}(t)\leq 1-\big(\tilde{v}(t)\big)^{2}\quad\forall t\in[t^{*},T]
\end{align*}
which can be solved explicitly to obtain
\begin{align*}
    \tilde{v}(t)\leq \tfrac{v(t^{*})\exp(2t)+v(t^{*})+\exp(2t)-1}{\exp(2t)+1+v(t^{*})(\exp(2t)-1)}.
\end{align*}
However, the right hand side goes to \(-\infty\) when \(t\rightarrow \tfrac{1}{2}\ln\Big(\tfrac{-1+v(t^{*})}{1+v(t^{*})}\Big)\) and thus, also the solution for \(v\) ceases to exist for \(t\geq t^{
*}+\tfrac{1}{2}\ln\Big(\tfrac{-1-v(t^{*})}{1+v(t^{*})}\Big)\).
This is illustrated in \cref{fig:speed_minus_infty}.
\begin{figure}
\centering
          \begin{tikzpicture}[scale=0.6]
          \begin{axis}[
            width=0.5\textwidth,
            xlabel={\(t\)},
            ylabel={\((\mathcolor{blue}{x(t)},\mathcolor{red}{\xl(t)})\)},
                 legend pos= {north west},
                 ymin=-1.6,ymax=9]
          \addplot[color=red,line width=2pt, mark=none]  table [y index=1, x index=0, col sep=comma] {Tikz_data/Example_3_3_1.csv};
          \addplot[color=blue,line width=2pt, mark=none]  table [y index=2, x index=0, col sep=comma] {Tikz_data/Example_3_3_1.csv};
            \addlegendentry{leader};
          \addlegendentry{follower};
          \end{axis}
          \end{tikzpicture}
          \begin{tikzpicture}[scale=0.6]
          \begin{axis}[
            width=0.5\textwidth,
            xlabel={\(t\)},
            ylabel={\((\mathcolor{blue}{v(t)},\mathcolor{red}{\vl(t)})\)},
                 legend pos={north west} ,
                 ymin=-5,ymax=1.1]
          \addplot[color=red,line width=2pt, mark=none]  table [y index=1, x index=0, col sep=comma] {Tikz_data/Example_3_3_2.csv};
          \addplot[color=blue,line width=2pt, mark=none]  table [y index=2, x index=0, col sep=comma] {Tikz_data/Example_3_3_2.csv};
          \addplot[color=black, domain=0:3, line width=2pt,dotted]{0};
                  \end{axis}
          \end{tikzpicture}
        \begin{tikzpicture}[scale=0.6]
          \begin{axis}[
            width=0.5\textwidth,
            xlabel={\(t\)},
            ylabel={\(\mathcolor{blue}{a(t)}\)},
                 legend pos={north west},
                 ymin=-20,ymax=3]
          \addplot[color=blue,line width=2pt, mark=none]  table [y index=2, x index=0, col sep=comma] {Tikz_data/Example_3_3_3.csv};
        \addplot[color=black, domain=0:3, line width=2pt,dotted]{0};
          \end{axis}
          \end{tikzpicture}
    \caption{Same parameters as in \cref{fig:example_negative_speed_1}, but \(\xlz=l+1<l+s_{0}\). \textbf{Left:} Positions of the vehicles;
    \textbf{Middle:} Velocities of the vehicles, \textbf{Right:} Acceleration of the follower. Velocity of the \textcolor{blue}{follower} diverging to \(-\infty\) around \(t\approx 1.06\). The solution ceases to exist for larger \(t\).}
    \label{fig:speed_minus_infty}
    \end{figure}
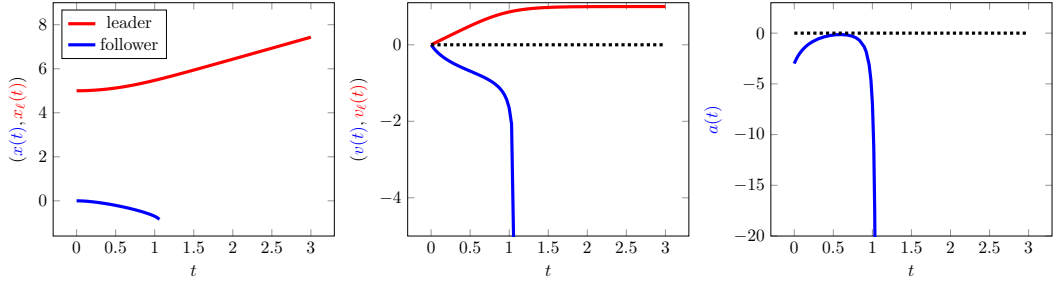
\end{example}

In this example, we chose very specific parameters for convenience. However, the same behavior and divergence of the speed to $-\infty$ in finite time can be shown for a rather general range of parameters. This is presented in \cref{ex:negative_velocity_blow_down}. To prove this, we start by showing the following Lemma.
\begin{lemma}[Sufficiently negative velocity in small time]
\label{lem_finite_blow_up}
Assume that the initial velocity of the following vehicle $v_{0}=0$, the initial velocity of the leading vehicle $\vlz>0$, and the initial positions of the two vehicles are separated by $l+\varepsilon$ for some $\varepsilon\in (0, s_0)$, i.e., $\xlz-x_{0}-l =\varepsilon<s_0$. Choose parameters $s_0$, $\tau$ and $v_{\text{free}}$ such that $-\frac{s_0}{1.01 \tau} < -v_{\text{free}}$, let $\delta>1$ 
and let $v_{\max}>0$ be an upper bound for the vehicles' velocities. Then there exists $t^{**} \in [0, T]$, such that $v(t^{**}) < -v_{\text{free}}$.
\end{lemma}
\begin{proof}
We consider two different cases: 
\begin{itemize}
\item If there exists $t^{*} \in [0, T]$, such that $v(t^{*}) <
-v_{\text{free}}$, then we chose $t^{**} = t^{*}$.
\item Otherwise, we assume that for every $t \in [0, T]$, $v(t) \geq -v_{\text{free}}$ which implies that $v(t)> -\frac{s_0}{1.01 \tau}$, i.e., 
\begin{equation}
\label{exm2_0}
s_0 + v(t) \tau >\tfrac{s_0}{101}>0.
\end{equation}
\end{itemize}
By the definition of $v_{\max}$ and $v_{\text{free}}$, we have for every $t \in [0, T]$, $\vl(t)-v(t) \leq v_{\max} + v_{\text{free}}$. Furthermore, since $\vl(t)>0$ and $v(t)<0$, we have $v(t)(v(t)-\vl(t))>0$. Then one can find an upper bound for the distance between the two vehicles.  That is, for every  $t \in [0,T]$, \begin{equation}
\label{upper_bounded_distance}
\xl(t)-x(t) -l \leq \varepsilon+(v_{\text{max}}+v_{\text{free}}) t.
\end{equation}
Therefore, for every $t \in [0, T]$, \begin{align}
\label{exm2_1}
\dot{v}(t) &= a\bigg(1-\big(\tfrac{|v(t)|}{v_{\text{free}}}\big)^{\delta}-\Big(\tfrac{2\sqrt{ab}(s_{0}+v(t)\tau)+v(t)(v(t)-\vl(t))}{2\sqrt{ab}(\xl(t)-x(t)-l)}\Big)^{2}\bigg)\\ \label{exm2_2}
    &\leq  a\bigg(1-\Big(\tfrac{2\sqrt{ab}(s_{0}+v(t)\tau)+v(t)(v(t)-\vl(t))}{2\sqrt{ab}(\xl(t)-x(t)-l)}\Big)^{2}\bigg)\\ \label{exm2_3}
    &\leq a\bigg(1-\Big(\tfrac{s_{0}+v(t)\tau}{\xl(t)-x(t)-l}\Big)^{2}\bigg)\\
    \label{exm2_4}
   &\leq a\bigg(1-\Big(\tfrac{s_{0}+v(t)\tau}{\varepsilon+(v_{\max}+v_{\text{free}})t}\Big)^{2}\bigg)\\ \label{exm2_5}
    &\leq a\bigg(1-\Big(\tfrac{\frac{s_0}{101}}{\varepsilon+(v_{\max}+v_{\text{free}})t}\Big)^{2}\bigg).\notag
\end{align}
Note that inequality \eqref{exm2_2} is true since $\frac{|v(t)|}{v_{\text{free}}}$ is non-negative, inequality  \eqref{exm2_3} is due to the fact that for every $t \in [0, T]$, $v(t)(v(t)-\vl(t)) >0$,
inequality \eqref{exm2_4} is because of the upper bounded for the distance between the two vehicles given by inequality \eqref{upper_bounded_distance}, and inequality \eqref{exm2_5} is due to inequality \eqref{exm2_0}.

Hence, for every $t \in [0,T]$, 
\begin{align*}
    v(t) &\leq a\int_0^{t} \bigg(1-\Big(\tfrac{\frac{s_0}{101}}{\varepsilon+(v_{\max}+v_{\text{free}})s}\Big)^{2}\bigg) \dd s \\
    &= a \bigg( t + \tfrac{s_0^2}{101^2 (v_{\max}+v_{\text{free}})}\left( \tfrac{1}{\varepsilon+(v_{\max}+v_{\text{free}})t} -\tfrac{1}{\varepsilon}\right)   \bigg)
    = a t\bigg(1-\tfrac{s_0^2}{101^2}\tfrac{1}{\varepsilon(\varepsilon+(v_{\max}+v_{\text{free}})t)}    \bigg).
\end{align*}
Setting $t = \varepsilon$, we have, 
\begin{align*}
      v( \varepsilon) &\leq  a  \varepsilon\bigg(1-\tfrac{s_0^2}{101^2}\tfrac{1}{\varepsilon(\varepsilon+(v_{\max}+v_{\text{free}}) \varepsilon)}    \bigg)= a\tfrac{101^2\varepsilon(\varepsilon+(v_{\text{max}}+v_{\text{free}})\varepsilon)-s_0^2}{101^2(\varepsilon+(v_{\text{max}}+v_{\text{free}})\varepsilon)}\rightarrow -\infty \text { as } \varepsilon \to 0. 
\end{align*}
Note that the upper bound for the velocity at time $t=\varepsilon$ goes to negative infinity as $\varepsilon$ goes to zero. 
Therefore, for $\varepsilon>0$ small enough, we have,
\begin{align*}
    v(\varepsilon) \leq -v_{\text{free}}.
\end{align*}
Hence, there exists $t^{**}=\varepsilon \in [0, T]$, such that $v(t^{**}) < -v_{\text{free}}$. 
\end{proof}
The previous Lemma guarantees that for properly chosen initial datum and velocity the follower's velocity can become more negative than \(-v_{0}\), the negative free-flow velocity. This enables us to prove that the solution ceases to exist in finite time. This result is related to the famous example of a blow-up of solutions to the Riccati ODEs for specific initial datum in finite time:
\begin{example}[Negative velocity and a blow-``down'' of the velocity in finite time]\label{ex:negative_velocity_blow_down} Again, we assume the parameters and initial data as in \cref{lem_finite_blow_up}. Then by \cref{lem_finite_blow_up}, there exists $t^{**} \in [0, T]$, such that $v(t^{**}) < -v_{\text{free}}$.
Recall that $\delta>1$. For any time $t$ such that $v(t) < -v_{\text{free}}$, 
\begin{align*}
   \dot{v}(t) \leq a\Big(1- \big(\tfrac{|v(t)|}{v_{\text{free}}}\big)^{\delta}\Big) <0
\end{align*}
Thus $v$ is strictly decreasing on the time interval $[t^{**}, T]$.  
In addition, we are going to show that $v$ ceases to exist in finite time and that there exists $t_{1}>0$ such that $\lim\limits_{t\rightarrow t_{1}} v(t) =-\infty$. Assume by contradiction that this is not the case. Then as $v(T)<-v_{0}$ the solution remains strictly decreasing as long as it exists and as it does not reach $-\infty$ in finite time it can be extended on $[0,+\infty)$ and is strictly decreasing on $[T,+\infty)$. Therefore we have the following
\begin{equation}
\int_{T}^{t}\tfrac{\dot v(s)}{1-\left(\frac{|v(s)|}{v_{\text{free}}}\right)^{\delta}}\dd s\geq a (t-T),\ \forall t\in [T,+\infty).\end{equation}
As $v$ is strictly decreasing we can perform a change of variable in the integral by setting $y = -v(s)$ to get
\begin{equation}\label{eq-interm}
\int_{-v(T)}^{-v(t)}\tfrac{1}{\left(\frac{y}{v_{\text{free}}}\right)^{\delta}-1}\dd y\geq a (t-T),\ \forall t\in [T,\infty).
\end{equation}
Note that $-v(T)>v_{0}$. We denote $\eta \coloneqq -v(T)-v_{\text{free}}>0$, \cref{eq-interm} implies
\begin{equation}
    \int_{v_{\text{free}}+\eta}^{+\infty}\tfrac{1}{\left(\frac{y}{v_{\text{free}}}\right)^{\delta}-1}\dd y\geq a(t-T),\ \forall t\in [T,\infty).
\end{equation}
Letting $t\rightarrow+\infty$ this implies that 
\begin{equation}
 \int_{v_{\text{free}}+\eta}^{+\infty}\tfrac{1}{\left(\frac{y}{v_{\text{free}}}\right)^{\delta}-1}\dd y = \infty,
\end{equation}
but because $\delta >1$ we have  \(\int_{v_{\text{free}}+\eta}^{+\infty}\tfrac{1}{\left(\frac{y}{v_{\text{free}}}\right)^{\delta}-1}\dd y\in \R\), which gives a contradiction. Therefore, $v$ ceases to exist and diverges to $-\infty$ in finite time.
\end{example}
As we have seen from the previous \cref{ex:negative_velocity_blow_down} the velocity can blow-up in finite time. However, what is not clear is whether the position of the car can consequently also explode. Thanks to the relation between position and velocity, i.e., \(x'(t)=v(t),\ t\in[0,T]\) this is a matter of whether \(v\in L^{1}((0,t^{*}))\) if \(t^{*}\) is the time where the velocity goes to \(-\infty\). And indeed, it can be shown that this holds true and the position remains bounded:
\begin{corollary}[Boundedness of the position in the case of a blow-up of velocity]
Let \cref{ass:input_datum}, \(\delta\in\R_{>2}\) and assume that -- as investigated in \cref{ex:negative_velocity_blow_down} -- there exists a time horizon \(t^{*}\in\R_{>0}\) so that
\[
\lim_{t\nearrow t^{*}} v(t)=\lim_{t\nearrow t^{*}} \dot{x}(t)=-\infty.
\]
Then, the position at the time of the blow-up remains finite, i.e.\ 
\[
\exists c\in\R: \lim_{t\nearrow t^{*}}x(t)=c
\]
or equivalently stated
\[
v\in L^{1}((0,t^{*})).
\]
\end{corollary}
\begin{proof}
The proof consists of showing that the \(L^{1}\) mass of the velocity remains bounded. To this end, we estimate the acceleration from 
above. 
Choose \(t_{1}\in [0,t^{*}]\) so that
\begin{align*}
\left(v(s)^{2}-v(s)\vl(s)+2\sqrt{ab} \big(s_{0}+v(s)\tau\big)\right)^{2}> 0 \wedge  v(s)\leq -\max\big\{ 2 v_{\text{free}},1\big\}\qquad \forall s\in[t_{1},t^{*}).
\end{align*}
Such a \(t_{1}\) always exists as \(v\) diverges to \(-\infty\) so that for \(s\) close enough to \(t^{*}\) the quadratic term in the previous estimate will always outnumber the affine linear term and \(\vl\) is essentially bounded. 
Then, recalling \cref{defi:IDM_acc} of the IDM we have for \(s\in [t_{1},t^{*})\)
\begin{equation*}
    \dot v(t) \leq -\tfrac{a}{2}\left(\tfrac{|v(t)|}{v_{\text{free}}}\right)^{\delta}.
\end{equation*}
Assuming \(\delta\in\R_{>2}\) divide by $|v(t)|^{\delta-1}$ (this is possible because $v$ is strictly decreasing and $v(t_{1})<-1$), and integrating between $t_{1}$ and $t\in[t_{1},t^{*})$, one has
\begin{equation}
    \begin{split}
            \int_{t_{1}}^{t} \tfrac{\dot v(\tau)}{|v(\tau)|^{\delta-1}}d\tau &\leq -\tfrac{a}{2v_{\text{free}}^{\delta}}\int_{t_{1}}^{t}|v(\tau)|\dd\tau,\\
            \tfrac{(-v)^{2-\delta}(t)}{\delta-2}-\tfrac{(-v)^{2-\delta}(t_{1})}{\delta-2} &\leq -\tfrac{a}{2v_{\text{free}}^{\delta}}\int_{t_{1}}^{t}|v(\tau)|\dd\tau.
    \end{split}
\end{equation}
Dividing by $-a/2v_{\text{free}}^{\delta}<0$ and letting $t\rightarrow t^{*}$, this gives
\begin{equation}
    \begin{split}
            \left(\tfrac{a}{2v_{\text{free}}^{\delta}}\right)^{-1}\tfrac{(-v)^{2-\delta}(t_{1})}{\delta-2} &\geq \|v\|_{L^{1}((t_{1},t^{*}))}.
    \end{split}
\end{equation}
Hence, $\|v\|_{L^{1}((t_{1},t^{*}))}<+\infty$.

\end{proof}
The previous estimate is particularly interesting as it illustrates that the model behaves still ``somehwat'' reasonable (even in the case of a diverge of the velocity to \(-\infty\)) and underlines the fact that a change in the acceleration to prevent the velocity to diverge might be enough to ``improve'' the model (compare \Cref{sec:IDM_improvements}).




\section{Lower bounds on the distance in specific cases}\label{sec:well_pos_spec_dat_times}
In this section we state results guaranteeing the minimal distance between leader and follower for the IDM.
\begin{theorem}[Minimal ``safety distance'' for the IDM]\label{theo:classical_IDM_safety_distance}
Let \cref{ass:input_datum} hold and particularly \(\vl\geqq 0\).
Assume that for an arbitrary time \(T\in\R_{>0}\) the solution to the IDM exists. Define the relative velocity of the leader and follower at time \(t \in [0, T]\) by \(\vl(t)-v(t)\).
Then, the IDM as in \cref{defi:IDM_model} satisfies the following lower bound on the distance
\begin{description}
\item[if the initial relative velocity is positive, i.e., \(\vlz-v_{0}\geq 0\)]
\begin{align}
    \xl(t)-x(t)-l\geq \min\bigg\{\xlz-x_{0}-l,\sqrt{\tfrac{as_{0}^{2}}{-B}}\bigg\} >0\quad \forall t\in[0,T]\label{eq:lower_bound_headway}
\end{align}
\item[if the initial relative velocity is negative, i.e., \(\vlz-v_{0}< 0\)]
\begin{align}
    \xl(t)-x(t)-l\geq \min\bigg\{\tfrac{-A+\sqrt{A^{2}+4aBs_{0}^{2}}}{2B},\sqrt{\tfrac{as_{0}^{2}}{-B}}\bigg\} >0\quad \forall t\in[0,T]\label{eq:lower_bound_headway_2}
\end{align}
\end{description}
with the constants \(A,B\) given as
\begin{equation}
    \begin{split}
    A&\:-B\cdot(\xlz-x_{0}-l)+a\tfrac{s_{0}^{2}}{\xlz-x_{0}-l}+\tfrac{1}{2}\big(\vlz-v_{0})^{2}>0\\
    B&\: \essinf_{s\in[0,T]}\ul(s)-a.
    \end{split}
    \label{defi:A_B}
\end{equation}
\end{theorem}
\begin{proof}
We start with considering the difference of the change between leader's acceleration and follower's acceleration to obtain for \(t\in[0,T]\)
\begin{align*}
    \ddot{\xl}(t)-\ddot{x}(t)&\geq \ul(t)-a+a\Big(\tfrac{|v(t)|}{v_{\text{free}}}\Big)^{\delta}+a\Big(\tfrac{2\sqrt{ab}(s_{0}+v(t)\tau)+v(t)(v(t)-\vl(t))}{2\sqrt{ab}\big(\xl(t)-x(t)-l\big)}\Big)^{2}\\
    &\geq\ul(t)-a+a\Big(\tfrac{2\sqrt{ab}(s_{0}+v(t)\tau)+v(t)(v(t)-\vl(t))}{2\sqrt{ab}\big(\xl(t)-x(t)-l\big)}\Big)^{2}.
\end{align*}
Let us first assume that \(\vlz-v_{0}< 0\). Then, we know that on a time horizon \([0,t_{1})\)\textcolor{blue}{,} \(\vl(t)-v(t)< 0\) and the distance of follower and leader decreases but there is still no over-taking, i.e., \(\xl(t)-x(t)-l>0\).
Thus, we can continue the previous estimate to arrive at (recall that \(\vl(t)\geqq0\) so that \(v(t)> 0\) as \(v(t)>\vl(t)\) for \(t\in [0,t_{1})\))
\begin{align*}
    \ddot{\xl}(t)-\ddot{x}(t)\geq \essinf_{s\in[0,T]}\ul(s)-a+as_{0}^{2}\tfrac{1}{(\xl(t)-x(t)-l)^{2}}.
\end{align*}
Multiplying with \(\dot{\xl}(t)-\dot{x}(t)< 0\) leads to
\begin{align*}
    \big(\ddot{\xl}(t)-\ddot{x}(t)\big)\big(\dot{\xl}(t)-\dot{x}(t)\big)\leq \Big(\essinf_{s\in[0,T]}\ul(s)-a\Big)\big(\dot{\xl}(t)-\dot{x}(t)\big)+as_{0}^{2}\tfrac{\dot{\xl}(t)-\dot{x}(t)}{(\xl(t)-x(t)-l)^{2}}
\end{align*}
and integrating over \(t\in [0,t_{1}]\) gives
\begin{align*}
    \tfrac{1}{2}(\vl(t)-v(t))^{2}& \leq \tfrac{1}{2}(\vlz-v_{0})^{2}+\Big(\essinf_{s\in[0,T]}\ul(s)-a\Big)\Big(\xl(t)-x(t)-\xlz+x_{0}\Big)\\
    &\qquad -as_{0}^{2}\Big(\tfrac{1}{\xl(t)-x(t)-l} -\tfrac{1}{\xlz-x_{0}-l}\Big)
\end{align*}    
Defining \(A\: -B\cdot(\xlz-x_{0}-l)+a\tfrac{s_{0}^{2}}{\xlz-x_{0}-l}+\tfrac{1}{2}\big(\vlz-v_{0})^{2}>0\) and \(g(t)\:\xl(t)-x(t)-l\) with \(B\: \essinf\limits_{s\in[0,T]}\ul(s)-a.\) we have
\begin{align*}
\tfrac{1}{2}\big(g'(t)\big)^{2}\leq A+Bg(t)-\tfrac{as_{0}^{2}}{g(t)}.
\end{align*}
However, as the left hand side is quadratic, the following inequality needs to hold (recall that $g(t)> 0$ on $[0,t_{1})$)
\[
0\leq Ag(t)+Bg(t)^{2}-as_{0}^{2}.
\]
Recalling that \(B<0\), we thus obtain as lower bound
\[
g^{*}\:\tfrac{-A+\sqrt{A^{2}+4aBs_{0}^{2}}}{2B}>0
\]
which is greater zero as \(A>0\) and \(aBs_{0}^{2}<0\) by assumption.

However, this is only a lower bound for the first time the relative velocity is negative as we needed to address the case where the initial relative velocity is negative. 
In the case where \(\vlz>v_{0}\) we will derive in the following a uniform lower bound.
In both cases, assume that there is another time \(t_{2},t_{3}\in (t_{1},T)\) so that
\(\vl(t)-v(t)<0\forall t\in(t_{2},t_{3}) \), we can assume that \(\vl(t_{2})=v(t_{2})\). Applying then the previous estimates once more, we obtain this time as lower bound
\[
\xl(t)-x(t)-l\geq \tfrac{-A_{2}+\sqrt{A^{2}_{2}+4aBs_{0}^{2}}}{2B}\ \forall t\in(t_{2},t_{3})
\]
with \(A_{2}\:-B\cdot(\xl(t_{2})-x(t_{2})-l)+a\tfrac{s_{0}^{2}}{\xl(t_{2})-x(t_{2})-l}>0\).
Looking into the discriminant we find that
\[
 A^{2}_{2}+4aBs_{0}^{2}=\Big(\tfrac{as_{0}^{2}}{\xl(t_{2})-x(t_{2})-l}+B(\xl(t_{2})-x(t_{2})-l)\Big)^{2}
\]
so that we obtain for \(t\in(t_{2},t_{3})\)
\[
\xl(t)-x(t)-l\geq \begin{cases}\xl(t_{2})-x(t_{2})-l & \text{ if }\xl(t_{2})-x(t_{2})-l\leq \sqrt{\tfrac{as_{0}^{2}}{-B}}\\
\tfrac{-a\tfrac{s_{0}^{2}}{\xl(t_{2})-x(t_{2})-l}}{B} & \text{ if }\xl(t_{2})-x(t_{2})-l\geq \sqrt{\tfrac{as_{0}^{2}}{-B}}
\end{cases}.
\]
Recalling that we can estimate from the previous (first) step and the fact that $\xl(t)-x(t)-l$ is non-decreasing between $t_{1}$ and $t_{2}$
\[
\xl(t_{2})-x(t_{2})-l\geq \min\big\{g^{*},\xlz-x_{0}-l\big\}
\]
we obtain with the previous estimate that for any \(t\in[0,t_{3}]\)
\[
\xl(t)-x(t)-l\geq \min\Big\{g^{*},\xlz-x_{0}-l,\sqrt{\tfrac{as_{0}^{2}}{-B}}\Big\}.
\]
However, this lower bound is independent of \(\xl(t_{2})-x(t_{2})-l\) and we can thus iterated the procedure by going to the next time where \(\vl(t)-v(t)<0\) for some \(t\in (t_{3},T]\). However, in these cases the previously derived bound remains as is. 

Looking into the derived lower bound in more details, one can actually distinguish the two cases \(\vlz-v_{0}\leq (\geq0)\) and arrives as the obtained bounds. This concludes the proof.
\end{proof}

\begin{remark}[Comments on the derived ``safety distance'' and the ``extreme'' case \(B=0\)]
We have not commented about the sign of \(A\) and \(B\) in \cref{defi:A_B}.
Clearly, assuming that \(B\leq 0\) is reasonable as otherwise it holds
\[
\essinf_{t\in[0,T]}\ul(s)\geq a,
\]
meaning that the leader speeds up all the time at least with the maximal acceleration of the follower, implying that the distance will always increase.
Thus, assuming \(B\leq 0\) the term \(\tfrac{-A}{2B}\) is positive and \(\frac{\sqrt{A^{2}+4aBs_{0}^{2}}}{2B}\) negative but its absolute values is smaller than \(\tfrac{-A}{2B}\) so that the obtained lower bound in \cref{eq:lower_bound_headway_2} is still positive.

The lower bound together with the corresponding simulations is illustrated for a specific experimental setup in \cref{fig:headway}.
As the lower bound is not well-defined for \(B=0\), we compute the limes of the lower bound. Recalling that \(A\) is also a function of \(B\) namely, \(A(B)= -B\cdot(\xlz-x_{0}-l)+a\tfrac{s_{0}^{2}}{\xlz-x_{0}-l}+\tfrac{1}{2}(\vlz-v_{0})^{2}\) we have
\begin{align*}
   \lim_{B\rightarrow 0} \tfrac{-A(B)+\sqrt{A(B)^{2}+4aBs_{0}^{2}}}{2B}&=\lim_{B\rightarrow 0}\tfrac{-A'(B)}{2}+\tfrac{A(B)A'(B)+2as_{0}^{2}}{2\sqrt{A(B)^{2}+4aBs_{0}^{2}}}\\
   &=\tfrac{as_{0}^{2}}{\tfrac{as_{0}^{2}}{\xlz-x_{0}-l}+\tfrac{1}{2}(\vlz-v_{0})^{2}}.
\end{align*}
Note that the obtained lower bound is always less or equal to the initial space headway, \(\xlz-x_{0}-l\). Furthermore, negative relative initial velocity leads to smaller lower bound.  But in the case of positive initial relative velocity, the given lower bound is very conservative and could be replaced by a stricter one. Altogether, even for \(B=0\) the obtained lower bound is reasonable.
\end{remark}

\section{Improvements for the IDM}\label{sec:IDM_improvements}
In this section, we present several improvements of the IDM to fix the problems illustrated in 
\Cref{sec:counterexamples} for general initial datum.
Before doing this, however, we present some other numerics on how the classical IDM behaves for specific data. This will serve as a comparison to the proposed improvements later:
\begin{example}[Some additional numerical results for the classical IDM]
All examples -- except those which are physically unreasonable (compare \Cref{subsec:projection_nonnegative_bounded_acceleration}) -- will be tested on three different scenarios:
\begin{enumerate}[leftmargin=15pt]
    \item As the set of parameters where one can observe a negative velocity of the follower in the original IDM -- see \cref{ex:negative_velocity} and \cref{fig:example_negative_speed_1}.\label{item:case_1}
    \item As the set of parameters where one can observe that the velocity of the follower diverges to \(-\infty\) in finite time in the original IDM -- see \cref{ex:negative_velocity_blowup} and \cref{fig:speed_minus_infty}. \label{item:case_2}
    \item A heavy stop and go wave traffic situation with the leader's acceleration satisfying \(u_{\text{lead}}\equiv a\cdot\mathds{1}_{\{t\in[0,T]:\ \sin(t/4)\geq 0.8\}} - a\cdot\mathds{1}_{\{t\in[0,T]:\ \sin(t/4)\leq -0.8\}}\),\  \(a=0.73,\ b=1.67,\ v_{\text{free}}=\frac{120}{36},\ \tau=1.6,\ l=4,\ s_{0}=2,\ d=4\) illustrated in \cref{fig:acc_leader}. The results for the original IDM are then illustrated in \cref{fig:original_IDM_stop_and_go} and can serve as comparison.
    \label{item:case_3}
\end{enumerate}
\begin{figure}
\centering
    \begin{tikzpicture}[scale=0.6]
        \begin{axis}[xscale=2,
            width=0.5\textwidth,  ylabel style={yshift=-1.5cm},
            xlabel={\(t\)},
            ylabel={\(u_{\text{lead}}(t)\)},
                 legend pos={north west} ]
        \addplot[color=blue,line width=1.5pt,color=red,domain=0:3.71] {0};
        \addplot[color=blue,line width=1.5pt,color=red,domain=0:0.73] (3.71,x);
        \addplot[color=blue,line width=1.5pt,color=red,domain=3.71:8.86] {0.73};
        \addplot[color=blue,line width=1.5pt,color=red,domain=0:0.73] (8.86,0.73-x);
        \addplot[color=blue,line width=1.5pt,color=red,domain=8.86:16.28] {0};
        \addplot[color=blue,line width=1.5pt,color=red,domain=0:0.73] (16.28,-x);
        \addplot[color=blue,line width=1.5pt,color=red,domain=16.28:21.43] {-0.73};
        \addplot[color=blue,line width=1.5pt,color=red,domain=0:0.73] (21.43,-0.73+x);
        \addplot[color=blue,line width=1.5pt,color=red,domain=21.43:28.85] {0};
        \addplot[color=blue,line width=1.5pt,color=red,domain=0:0.73] (28.85,x);
        
        \addplot[color=blue,line width=1.5pt,color=red,domain=28.85:34.00] {0.73};
        \addplot[color=blue,line width=1.5pt,color=red,domain=0:0.73] (34.00,0.73-x);
        \addplot[color=blue,line width=1.5pt,color=red,domain=34.00:41.42] {0};
        \addplot[color=blue,line width=1.5pt,color=red,domain=0:0.73] (41.42,-x);
        
        \addplot[color=blue,line width=1.5pt,color=red,domain=41.42:46.57] {-0.73};
        \addplot[color=blue,line width=1.5pt,color=red,domain=0:0.73] (46.57,-0.73+x);
        \addplot[color=blue,line width=1.5pt,color=red,domain=46.57:53.99] {0};
        \addplot[color=blue,line width=1.5pt,color=red,domain=0:0.73] (53.99,x);
        
        \addplot[color=blue,line width=1.5pt,color=red,domain=53.99:59.14] {0.73};
        \addplot[color=blue,line width=1.5pt,color=red,domain=0:0.73] (59.14,0.73-x);
        \addplot[color=blue,line width=1.5pt,color=red,domain=59.14:66.56] {0};
        \addplot[color=blue,line width=1.5pt,color=red,domain=0:0.73] (66.56,-x);
        
        \addplot[color=blue,line width=1.5pt,color=red,domain=66.56:71.71] {-0.73};
        \addplot[color=blue,line width=1.5pt,color=red,domain=0:0.73] (71.71,-0.73+x);
        \addplot[color=blue,line width=1.5pt,color=red,domain=71.71:79.13] {0};
        \addplot[color=blue,line width=1.5pt,color=red,domain=0:0.73] (79.13,x);
        
        \addplot[color=blue,line width=1.5pt,color=red,domain=79.13:84.28] {0.73};
        \addplot[color=blue,line width=1.5pt,color=red,domain=0:0.73] (84.28,0.73-x);
        \addplot[color=blue,line width=1.5pt,color=red,domain=84.28:91.70] {0};
        \addplot[color=blue,line width=1.5pt,color=red,domain=0:0.73] (91.70,-x);
        
        \addplot[color=blue,line width=1.5pt,color=red,domain=91.70:96.85] {-0.73};
        \addplot[color=blue,line width=1.5pt,color=red,domain=0:0.73] (96.85,-0.73+x);
        \addplot[color=blue,line width=1.5pt,color=red,domain=96.85:100] {0};
\end{axis}
    \end{tikzpicture}
    \caption{The acceleration profile \(u_{\text{lead}}\) of the leader. This is intended to specify a leader that repeats the pattern ``accelerate, constant velocity, decelerate, constant velocity''.}
    \label{fig:acc_leader}
\end{figure}
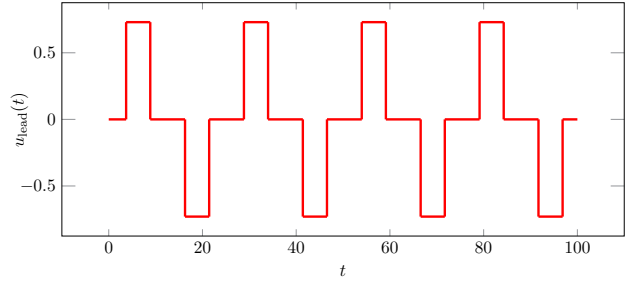

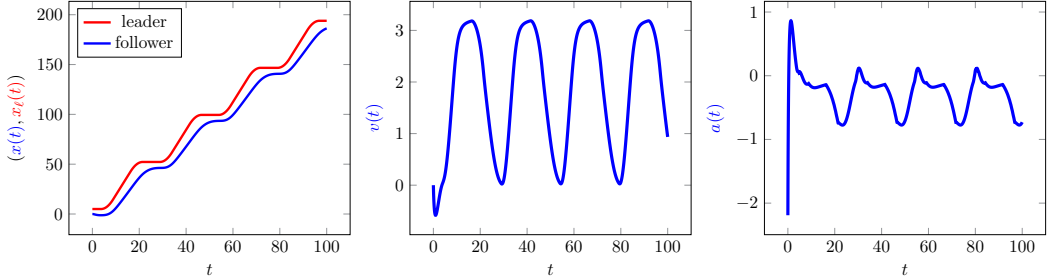
\begin{figure}
\centering
    \begin{tikzpicture}[scale=0.6]
        \begin{axis}[
            width=0.5\textwidth,
            xlabel={\(t\)},
            ylabel={\((\mathcolor{blue}{x(t)},\mathcolor{red}{\xl(t)})\)},
                 legend pos={north west}       ]
        \addplot[color=red,line width=1.5pt, mark=none]  table [y index=1, x index=0, col sep=comma] {Tikz_data/IDM_positions.csv};
        \addplot[color=blue,line width=1.5pt, mark=none]  table [y index=2, x index=0, col sep=comma] {Tikz_data/IDM_positions.csv};
        \addlegendentry{leader};
        \addlegendentry{follower};
        \end{axis}
    \end{tikzpicture}
      \begin{tikzpicture}[scale=0.6]
          \begin{axis}[
            width=0.5\textwidth,
            xlabel={\(t\)},
            ylabel={\(\mathcolor{blue}{v(t)}\)},
                 legend pos={north west}       ]
          \addplot[color=blue,line width=2pt, mark=none]  table [y index=1, x index=0, col sep=comma] {Tikz_data/IDM_velocity.csv};
          \end{axis}
          \end{tikzpicture}
     \begin{tikzpicture}[scale=0.6]
          \begin{axis}[
            width=0.5\textwidth,
            xlabel={\(t\)},
            ylabel={\(\mathcolor{blue}{a(t)}\)},
                 legend pos={north west}       ]
          \addplot[color=blue,line width=2pt, mark=none]  table [y index=1, x index=0, col sep=comma] {Tikz_data/IDM_acc.csv};
          \end{axis}
          \end{tikzpicture}
        \caption{The original IDM, with parameters \(x_{0}=0,\ \xlz=l+\textcolor{red}{1}<l+s_{0},\ v_{0}=0=\vlz=0\). Both vehicles start with \(0\) velocity, and the leader follows the acceleration profile in \cref{fig:acc_leader}. 
        \textbf{Left:} the position of both \textcolor{red}{leader} and \textcolor{blue}{follower}, \textbf{middle:} the velocity of the \textcolor{blue}{follower} and \textbf{right:} the acceleration of the \textcolor{blue}{follower}. We leave the velocity and acceleration profile of the leader out as it is fully determined by the given \(u_{\text{lead}}\).} 
                \label{fig:original_IDM_stop_and_go}
        \end{figure}
\begin{figure}
\centering
    \begin{tikzpicture}[scale=0.6]
        \begin{axis}[
            width=0.5\textwidth,
            xlabel={\(t\)},
            ylabel={\(\xl(t)-x(t)-l\)},
                 legend pos={north west}       ]
        \addplot[color=magenta,line width=1.5pt, mark=none]  table [y index=3, x index=0, col sep=comma] {Tikz_data/IDM_positions.csv};
        \addplot[color=yellow!50!orange,dotted,line width=1.5pt, mark=none, domain=0:100] {1};
        \end{axis}
    \end{tikzpicture}
        \caption{The original IDM with the same parameters as in \cref{fig:original_IDM_stop_and_go}. Headway (in \textcolor{magenta}{magenta}) \(\xl(t)-x(t)-l\) is illustrated together with the lower ``a priori'' bound as derived in \cref{eq:lower_bound_headway} in \cref{theo:classical_IDM_safety_distance}. Computing this headway, we have for the numbers \(A,B\) as in \cref{defi:A_B} \(B=-0.73 -0.73=-1.46, \) \(A=1.46\cdot 1+ 0.73\tfrac{4}{1}=4.38\) and \(s_{0}=2,\) and thus as lower bound on the headway is -- following \cref{eq:lower_bound_headway} (recall that \(\vlz-v_{0}=0\) so that this case applied)
        \(\min\bigg\{\xlz-x_{0}-l,\sqrt{\tfrac{as_{0}^{2}}{-B}}\bigg\}=1\) which is pictured in \textcolor{yellow!50!orange}{dotted yellow}. As can be seen the lower bound is always valid and for \(t=0\) even sharp.} 
                \label{fig:headway}
        \end{figure}
\end{example}
\subsection{Projection on nonnegative velocities and restricting the maximal deceleration}
A straight forward improvement consists of projecting the velocity to nonnegative values. This is detailed in the following \cref{defi:IDM_projected_velocities}:
\begin{definition}[IDM with projection to nonnegative
velocities]\label{defi:IDM_projected_velocities}
Given \cref{ass:input_datum}, we replace the acceleration in \cref{defi:IDM_acc} and velocity for the IDM  model in \cref{defi:IDM_model} by
\begin{align}
    \dot{x}(t) &= \max\{v(t),0\},&& t\in[0,T]\\\label{eq:4242}
    \dot{v}(t) &= \Acc\big(x(t),\max\{v(t),0\},\xl(t),\vl(t)\big), && t\in[0,T]
\end{align}
and call the model the \textbf{velocity projected IDM}.
\end{definition}
\begin{theorem}[Existence and uniqueness of solutions for small times]
Given \cref{ass:input_datum} the velocity projected IDM in \cref{defi:IDM_projected_velocities} admits on a sufficiently small time horizon \(T^{*}\in\R_{>0}\) a unique solution \((x,v)\in W^{1,\infty}((0,T^{*}))^{2}\).
\end{theorem}
\begin{proof}
The proof is almost identical to the proof of \cref{theo:well_posedness_small_time_horizon} when recalling that the right hand side is still locally Lipschitz-continuous. We do not go into details.
\end{proof}
However, as we will see the model has some drawbacks:
\begin{example}[The following vehicle waits for too long to start driving] Assume that the constants and initial data as in \cref{ex:negative_velocity} with $\varepsilon < s_0$. Then, the leading vehicle's trajectory can be computed as
\begin{align*}
    \xl(t)= \xlz+\vlz t+\tfrac{u}{2}t^{2}, \ t\in [0, T],
\end{align*}
and the leading vehicle's velocity can be computed as 
\begin{align*}
    \vl(t) = \vlz + u t, \ t\in [0, T].
\end{align*}
Plugging this into the change of the vehicle's velocity we obtain the following system of ODEs
\begin{align*}
    \dot{x}(t)&=\max\{v(t),0\}\\
    \dot{v}(t)&=a\bigg(1- \Big(\tfrac{\max\{v(t),0\}}{v_{\text{free}}}\Big)^{\delta}-\Big(\tfrac{2\sqrt{ab}(s_{0}+\max\{v(t),0\}\tau)+\max\{v(t),0\}(\max\{v(t),0\}-\vl(t))}{2\sqrt{ab}( \xl(t)-\xlz+\epsilon-\int_{0}^{t}\max\{v(s),0\}\dd s)}\Big)^{2}\bigg)\\
    x(0)&= \xlz-l-\epsilon\\
    v(0)&=0.    
\end{align*}
Note that
\begin{align*}
    \dot{v}(0) = a\left(1-\left(\tfrac{s_0}{\varepsilon}\right)^2\right)<0,
\end{align*}
therefore, there exists some small time interval $[0, t_1]$ such that for every $t \in [0, t_1]$, $v(t) <0$. During the time interval $[0, t_1]$, the distance between the two vehicles is 
\[l+\varepsilon+\xl(t)-\xlz = l+\varepsilon+\vlz t+\tfrac{u}{2}t^2,\quad t \in [0, t_1].\]
\end{example}
Thus, for every $t\in [0, t_1]$, 
\begin{align*}
    \dot{v}(t)=a\bigg(1 -\Big(\tfrac{s_{0}}{\varepsilon+\vlz t +\tfrac{u}{2}t^2}\Big)^{2}\bigg). 
\end{align*}
Note that $\dot{v} \colon [0, t_1] \mapsto \mathbb{R}$ is strictly increasing. 
Without loss of generality, we assume that the initial velocity of the leading vehicle is $\vlz=0$ and the acceleration of the leading vehicle is $u=2$. Then for every $t \in [0, t_1]$, 
\begin{align*}
    v(t)=&\int_0^t  a\Big(1 -\Big(\tfrac{s_{0}}{\varepsilon +s^2}\Big)^{2}\Big)  \dd s = at - 4as_0^2 \int_{0}^{t} \left(\tfrac{1}{s^2+\varepsilon} \right)^2\dd s
    = at - 4as_0^2\left(\tfrac{\frac{\sqrt{\varepsilon}t}{\varepsilon+t^2}+\arctan\left(\frac{t}{\sqrt{\varepsilon}}\right)}{2\varepsilon^{\frac{3}{2}}}\right).
\end{align*}
In particular, as illustrated in \cref{fig:Example5.3}, we have that $t_1$ increases as $\varepsilon$ decreases. That is, the smaller the initial distance between the two vehicles, the longer it takes the following vehicle to recover its positive velocity. 
\begin{figure}
    \centering
    \includegraphics[width=0.5\textwidth]{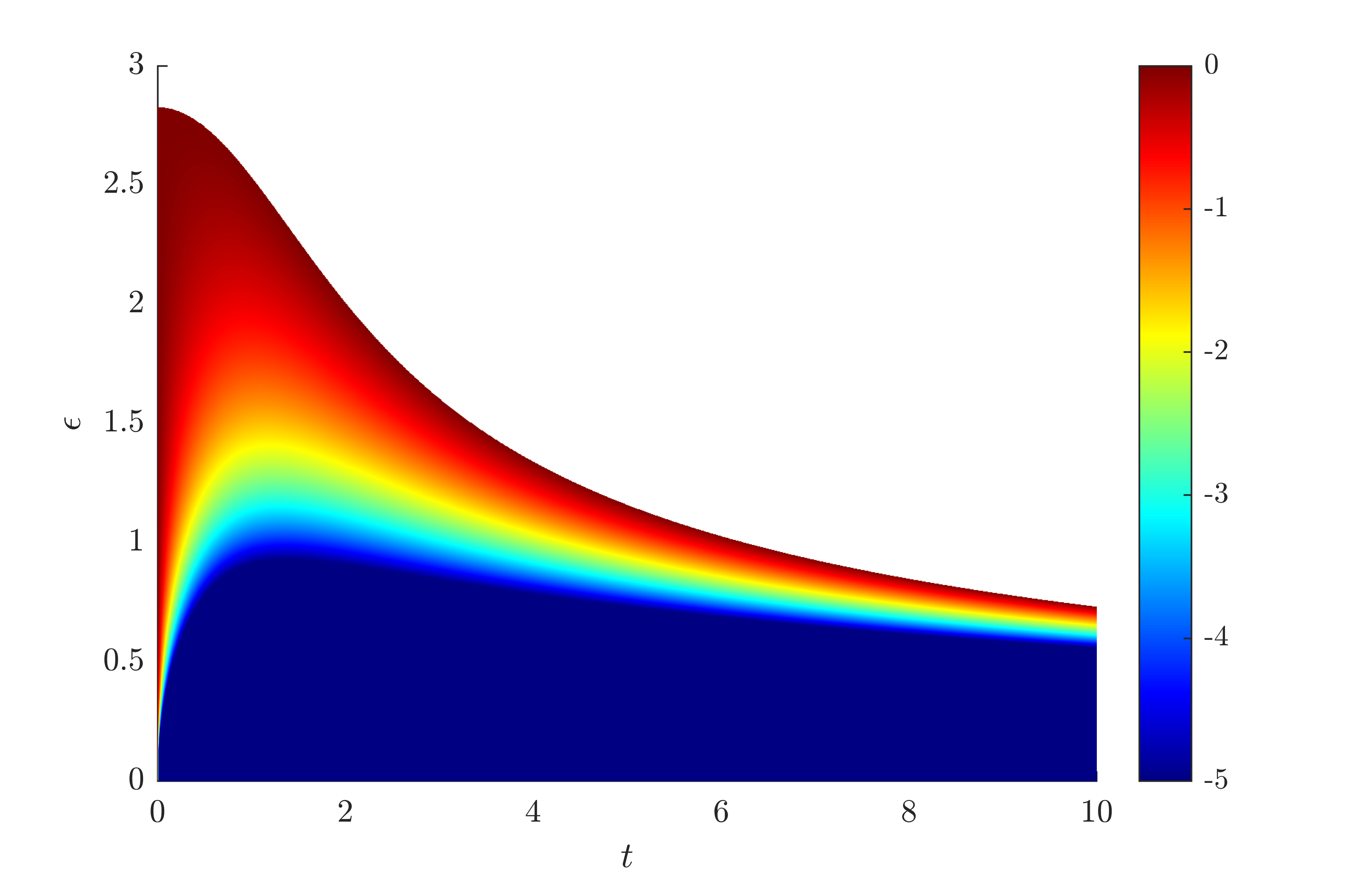}
    \caption{Velocity dependent on epsilon and time. The dark blue indicates values less or equal \(-5\) and the white area positive function values, so that the red curve separating the white and colored region can be seen as the values where the velocity is actually zero. In particular, as the initial distance between two vehicles $\varepsilon$ increases, the time when the following vehicle recovers its positive velocity $t_1$ decreases.}
    \label{fig:Example5.3}
\end{figure}
Another example illustrates the projected velocity model numerically with regard to other scenarios:
\begin{example}[Velocity projected IDM]
As can be observed the actual velocity in all the three different scenarios is bounded from below by zero and the solution exists on the entire time horizon considered. However, the projection operator leads to the problem that the follower waits too long until they speed up. This can be observed in particular in \cref{fig:projected_IDM_large_spacing_0,fig:projected_IDM_large_spacing} where the distance of the two vehicles after both have started speeding up (\(t\approx 5\)) is approximately around \(8.5\) which is quite far from the comfortable vehicle distance \(s_{0}\) and thus leading to a too large distance. Here we use the free-flow acceleration as following 
\begin{equation}
\dot{v}(t) = a-a \Big(\tfrac{\max\{v(t),0\}}{v_{\text{free}}}\Big)^{\delta}\ \forall t \in [0, T]. 
\label{free_flow_dynamics}
\end{equation} Same can be observed in \cref{fig:projected_IDM_large_spacing_2} for smaller time.
\begin{figure}
\centering
    \begin{tikzpicture}[scale=0.6]
        \begin{axis}[
            width=0.5\textwidth,
            xlabel={\(t\)},
            ylabel={\((\mathcolor{blue}{x(t)},\mathcolor{red}{\xl(t)})\)},
                 legend pos={north west}       ]
        \addplot[color=red,line width=1.5pt, mark=none]  table [y index=1, x index=0, col sep=comma] {Tikz_data/IDM_m1_positions_negative_v.csv};
        \addplot[color=blue,line width=1.5pt, mark=none]  table [y index=2, x index=0, col sep=comma] {Tikz_data/IDM_m1_positions_negative_v.csv};
        \addlegendentry{leader};
        \addlegendentry{follower};
        \end{axis}
    \end{tikzpicture}
      \begin{tikzpicture}[scale=0.6]
          \begin{axis}[
            width=0.5\textwidth,
            xlabel={\(t\)},
            ylabel={\(\mathcolor{blue}{v(t)}\)},
                 legend pos={north west}       ]
          \addplot[color=blue,line width=2pt, mark=none]  table [y index=1, x index=0, col sep=comma] {Tikz_data/IDM_m1_velocity_negative_v.csv};
          \end{axis}
          \end{tikzpicture}
     \begin{tikzpicture}[scale=0.6]
          \begin{axis}[
            width=0.5\textwidth,
            xlabel={\(t\)},
            ylabel={\(\mathcolor{blue}{a(t)}\)},
                 legend pos={north west}       ]
          \addplot[color=blue,line width=2pt, mark=none]  table [y index=1, x index=0, col sep=comma] {Tikz_data/IDM_m1_acc_negative_v.csv};
          \end{axis}
          \end{tikzpicture}
        \caption{The IDM with projection as in \cref{defi:IDM_projected_velocities} and same parameters as in \cref{fig:example_negative_speed_1} with \(x_{0}=0,\ \xlz=l+1.5<l+s_{0}\). \textbf{Left:} the positions of the vehicles, \textbf{middle:} the velocity of the \textcolor{blue}{follower}, and \textbf{right:} the acceleration of the \textcolor{blue}{follower}. The leader follows the free-flow acceleration as in \cref{free_flow_dynamics}. The follower stays still and waits until there is a safe space to speed up.}
        \label{fig:projected_IDM_large_spacing_0}
\end{figure}
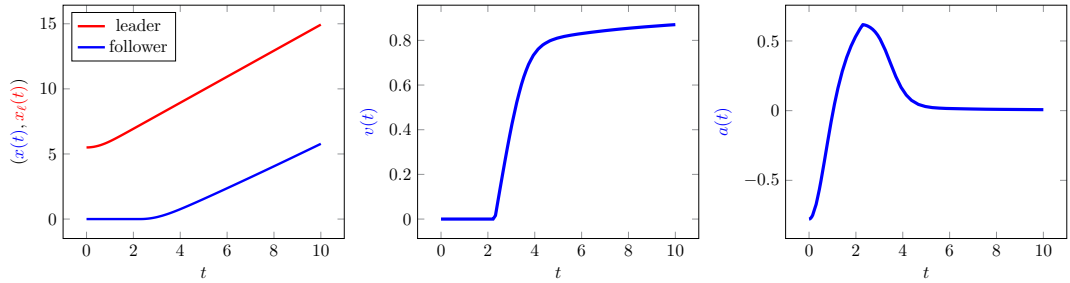

\begin{figure}
\centering
    \begin{tikzpicture}[scale=0.6]
        \begin{axis}[
            width=0.5\textwidth,
            xlabel={\(t\)},
            ylabel={\((\mathcolor{blue}{x(t)},\mathcolor{red}{\xl(t)})\)},
                 legend pos={north west}       ]
        \addplot[color=red,line width=1.5pt, mark=none]  table [y index=1, x index=0, col sep=comma] {Tikz_data/IDM_m1_positions_inf_v.csv};
        \addplot[color=blue,line width=1.5pt, mark=none]  table [y index=2, x index=0, col sep=comma] {Tikz_data/IDM_m1_positions_inf_v.csv};
        \addlegendentry{leader};
        \addlegendentry{follower};
        \end{axis}
    \end{tikzpicture}
      \begin{tikzpicture}[scale=0.6]
          \begin{axis}[
            width=0.5\textwidth,
            xlabel={\(t\)},
            ylabel={\(\mathcolor{blue}{v(t)}\)},
                 legend pos={north west}       ]
          \addplot[color=blue,line width=2pt, mark=none]  table [y index=1, x index=0, col sep=comma] {Tikz_data/IDM_m1_velocity_inf_v.csv};
          \end{axis}
          \end{tikzpicture}
     \begin{tikzpicture}[scale=0.6]
          \begin{axis}[
            width=0.5\textwidth,
            xlabel={\(t\)},
            ylabel={\(\mathcolor{blue}{a(t)}\)},
                 legend pos={north west}       ]
          \addplot[color=blue,line width=2pt, mark=none]  table [y index=1, x index=0, col sep=comma] {Tikz_data/IDM_m1_acc_inf_v.csv};
          \end{axis}
          \end{tikzpicture}
        \caption{The IDM with projection, with the same parameters as in \cref{fig:example_negative_speed_1} and \(x_{0}=0,\ \xlz=l+0.5<l+s_{0},\ v_{0}=0=\vlz=0\). \textbf{Left:} vehicles' positions, \textbf{middle:} the \textcolor{blue}{follower's} velocity and \textbf{right:} the \textcolor{blue}{follower's} acceleration. The leader follows the free flow acceleration as in \cref{free_flow_dynamics}.  The follower stays still and waits until there is a safe space to speed up.     }\label{fig:projected_IDM_large_spacing}
\end{figure}
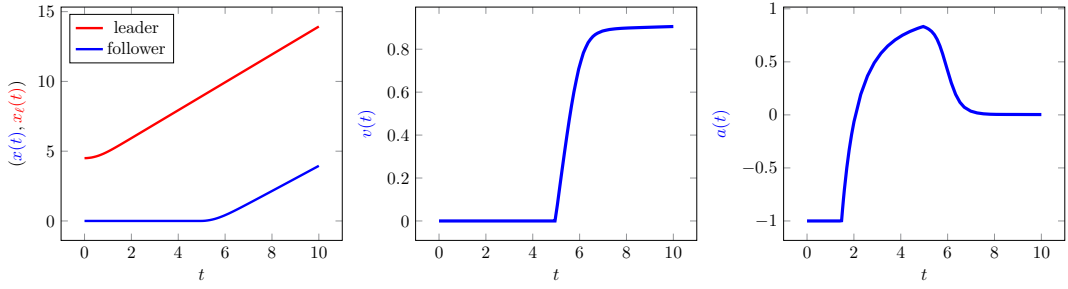

\begin{figure}
\centering
    \begin{tikzpicture}[scale=0.6]
        \begin{axis}[
            width=0.5\textwidth,
            xlabel={\(t\)},
            ylabel={\((\mathcolor{blue}{x(t)},\mathcolor{red}{\xl(t)})\)},
                 legend pos={north west}       ]
        \addplot[color=red,line width=1.5pt, mark=none]  table [y index=1, x index=0, col sep=comma] {Tikz_data/IDM_m1_positions.csv};
        \addplot[color=blue,line width=1.5pt, mark=none]  table [y index=2, x index=0, col sep=comma] {Tikz_data/IDM_m1_positions.csv};
        \addlegendentry{leader};
        \addlegendentry{follower};
        \end{axis}
    \end{tikzpicture}
      \begin{tikzpicture}[scale=0.6]
          \begin{axis}[
            width=0.5\textwidth,
            xlabel={\(t\)},
            ylabel={\(\mathcolor{blue}{v(t)}\)},
                 legend pos={north west}       ]
          \addplot[color=blue,line width=2pt, mark=none]  table [y index=1, x index=0, col sep=comma] {Tikz_data/IDM_m1_velocity.csv};
          \end{axis}
          \end{tikzpicture}
     \begin{tikzpicture}[scale=0.6]
          \begin{axis}[
            width=0.5\textwidth,
            xlabel={\(t\)},
            ylabel={\(\mathcolor{blue}{a(t)}\)},
                 legend pos={north west}       ]
          \addplot[color=blue,line width=2pt, mark=none]  table [y index=1, x index=0, col sep=comma] {Tikz_data/IDM_m1_acc.csv};
          \end{axis}
          \end{tikzpicture}
        \caption{The IDM with projection as in \cref{defi:IDM_projected_velocities} and initial data \(x_{0}=0,\ \xlz=l+\textcolor{red}{1}<l+s_{0},\ v_{0}=\vlz=0\). \textbf{Left:} the position of the vehicles, \textbf{middle:} the velocity of the \textcolor{blue}{follower} and \textbf{right:} the acceleration of the follower.
    The \textcolor{red}{leader} follows the acceleration profile in \cref{fig:acc_leader}. The velocities remain positive, however the waiting time for the follower until speeding up is large.}
    \label{fig:projected_IDM_large_spacing_2}
\end{figure}
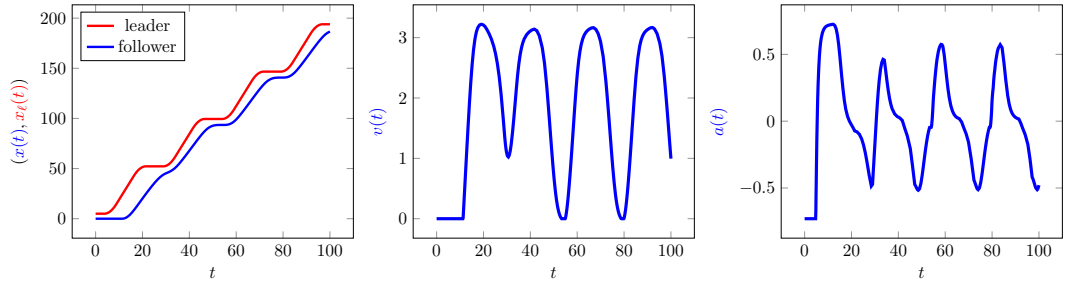

\end{example}
\subsection{Projection to nonnegative velocities with bounded deceleration}\label{subsec:projection_nonnegative_bounded_acceleration}
Another improvement for this is projecting the acceleration to prevent it from becoming too negative.
Then, the corresponding model reads as
\begin{definition}[IDM with projection to nonnegative velocities and bounded deceleration]\label{defi:IDM_projected_velocities_acceleration}
Given \cref{ass:input_datum}, we replace the acceleration in \cref{defi:IDM_acc} and velocity for the IDM  model in \cref{defi:IDM_model} by
\begin{align*}
    \dot{x}(t) &= \max\{v(t),0\},&& t\in[0,T]\\
    \dot{v}(t) &= \max\{\Acc(x(t),\max\{v(t),0\},\xl(t),\vl(t)),-a_{\min}\}, && t\in[0,T]
\end{align*}
with a parameter \(a_{\min}\in\R_{>0}\) be given and call the model the \textbf{acceleration projected IDM}.
\end{definition}
\begin{theorem}[Global existence and uniqueness of solutions]\label{theo:existence_acceleration_projected_IDM}
Given \cref{ass:input_datum} the acceleration projected IDM in \cref{defi:IDM_projected_velocities_acceleration} admits for every \(T\in\R_{>0}\) a unique solution \((x,v)\in W^{1,\infty}((0,T))\).
\end{theorem}
\begin{proof}
The proof of existence and uniqueness for small time is almost identical to the proof of \cref{theo:well_posedness_small_time_horizon} when recalling that the right hand side is still locally Lipschitz-continuous. We do not go into details.

So it remains to show that we can find uniform estimates for \((x(t),v(t)),\ t\in[0,T]\).
Obviously, \[v(t)\geq v_{0}-a_{\min}t\quad \forall t\in[0,T].\] Thanks to the structure of \(\Acc\) (see \cref{defi:IDM_acc}) we also obtain as a bound from above
\[v(t)\leq v_{0}+at\quad \forall t\in[0,T].\]
As \(v\) is uniformly bounded on every finite time horizon, so is \(x\) and we are done.
\end{proof}
However, although the previous change of the acceleration profile in \cref{defi:IDM_projected_velocities_acceleration} looks promising as according to \cref{theo:existence_acceleration_projected_IDM} a solution exists on every finite time horizon, the physical representation, the model itself is unreasonable as the car behind can overtake the leading car -- or differently put, the car behind can bump into the leading car without the model noticing it. This is detailed in the following \cref{exa:acceleration_projected_IDM_unreasonable}.
\begin{example}[Physical unreasonability]\label{exa:acceleration_projected_IDM_unreasonable}
As the deceleration of the following vehicle is bounded from below by \(-a_{\min},\) we can always chose an initial velocity of the follower which leads to the fact that \(\xl(t)-x(t)-l\rightarrow 0\) in finite time.

In formulae, assume for simplicity that the leading vehicle has the following trajectory
\[
\xl(t)={\xl}_{0}+{\vl}_{0}t + \tfrac{1}{2}u_{\text{lead}}t^{2} ,\qquad t\in[0,T].
\]
with \(u_{\text{lead}}\in\R_{\geq0}\) and \(({\xl}_{0},{\vl}_{0})\in \R\times\R_{>0}\). Then, we take the difference of the vehicles position with car length \(l\in\R_{>0}\) and have for \(t\in[0,T]\)
\begin{align*}
    &\xl(t)-x(t)-l\\
    &\leq {\xl}_{0}+{\vl}_{0}t-l + \tfrac{1}{2}u_{\text{lead}}t^{2}-x_{0}-\int_{0}^{t}\max\{v(s),0\}\dd s\\
    &\leq {\xl}_{0}+{\vl}_{0}t-l + \tfrac{1}{2}u_{\text{lead}}t^{2}-x_{0}-\int_{0}^{t}v(s)\dd s\\
    &\leq {\xl}_{0}+{\vl}_{0}t-l + \tfrac{1}{2}u_{\text{lead}}t^{2}-x_{0}\\
    &\qquad -\int_{0}^{t}v_{0}+\int_{0}^{s} \max\{\Acc(x(\tau),\max\{0,v(\tau)\},\xl(\tau),\vl(\tau)),-a_{\min}\}\dd\tau\dd s\\
    &\leq{\xl}_{0}+{\vl}_{0}t-l + \tfrac{1}{2}u_{\text{lead}}t^{2}-x_{0}-tv_{0}+\int_{0}^{t}\int_{0}^{s} a_{\min}\dd\tau\dd s\\
    &\leq {\xl}_{0}-x_{0}-l +({\vl}_{0}-v_{0})t +\tfrac{1}{2}(u_{\text{lead}}+a_{\min})t^{2}.
\end{align*}
Obviously, for \(v_{0}\) sufficiently large, we obtain for small time that \(\xl(t)-x(t)<l\).
An extreme case for this is illustrated in \cref{fig:follower_leader_intersecting} where not only the vehicles get closer than \(l\) but the follower (in \textcolor{blue}{blue}) entirely overtakes the leader (in \textcolor{red}{red}).
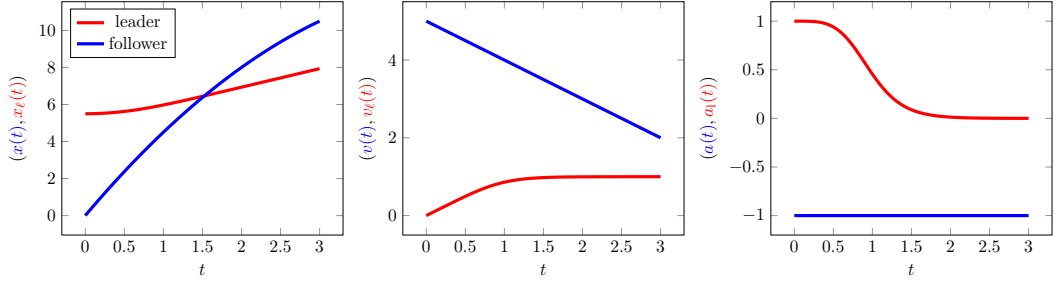
\begin{figure}
\centering
    \begin{tikzpicture}[scale=0.6]
        \begin{axis}[
            width=0.5\textwidth,
            xlabel={\(t\)},
            ylabel={\((\mathcolor{blue}{x(t)},\mathcolor{red}{\xl(t)})\)},
                 legend pos={north west}       ]
        \addplot[color=red,line width=2pt, mark=none]  table [y index=1, x index=0, col sep=comma] {Tikz_data/Example_6_2_1.csv};
        \addplot[color=blue,line width=2pt, mark=none]  table [y index=2, x index=0, col sep=comma] {Tikz_data/Example_6_2_1.csv};
        \addlegendentry{leader};
        \addlegendentry{follower};
        \end{axis}
    \end{tikzpicture}
    \begin{tikzpicture}[scale=0.6]
          \begin{axis}[
            width=0.5\textwidth,
            xlabel={\(t\)},
            ylabel={\((\mathcolor{blue}{v(t)},\mathcolor{red}{\vl(t)})\)},
                 legend pos={north west}       ]
          \addplot[color=red,line width=2pt, mark=none]  table [y index=1, x index=0, col sep=comma] {Tikz_data/Example_6_2_2.csv};
          \addplot[color=blue,line width=2pt, mark=none]  table [y index=2, x index=0, col sep=comma] {Tikz_data/Example_6_2_2.csv};
          \end{axis}
    \end{tikzpicture}
    \begin{tikzpicture}[scale=0.6]
          \begin{axis}[
            width=0.5\textwidth,
            xlabel={\(t\)},
            ylabel={\((\mathcolor{blue}{a(t)},\mathcolor{red}{a_{\text{l}}(t)})\)},
                 legend pos={north west}       ]
          \addplot[color=red,line width=2pt, mark=none]  table [y index=1, x index=0, col sep=comma] {Tikz_data/Example_6_2_3.csv};
          \addplot[color=blue,line width=2pt, mark=none]  table [y index=2, x index=0, col sep=comma] {Tikz_data/Example_6_2_3.csv};
          \end{axis}
    \end{tikzpicture}
        \caption{Illustration of the improvement in \cref{defi:IDM_projected_velocities_acceleration} and its physical unreasonability as shown in \cref{exa:acceleration_projected_IDM_unreasonable}. The parameters are \(a=a_{\min}=1,\ b=2,\ v_{\text{free}}=1,\ \tau=1.6,\ l=4,\ s_{0}=2,\ d=4\) and datum \(x_{0}=0,\ \xlz=l+1.5,\ v_0=5,\ \vlz=0 \).
        \textbf{Left} Vehicles' positions, \textbf{middle} vehicles' velocities, \textbf{right} vehicles' acceleration. 
        The \textcolor{blue}{follower} overtakes the \textcolor{red}{leader} in finite time, although the follower breaks constantly.}
    \label{fig:follower_leader_intersecting}
\end{figure}
\end{example}

\subsection{Another velocity projection improvement}
Another improvement which had been mentioned in the literature in \cite{milanes2014modeling} (however, it goes back to a website of Martin Treiber which is not available anymore, is shortly investigated in this section. Not the entire third part in the IDM acceleration \cref{defi:IDM_acc} is projected\textcolor{blue}{, but} only a specific part. This is detailed as follows:
\begin{definition}[Partially projected acceleration]
Let \cref{ass:input_datum} and \(\Acc\) as in \cref{defi:IDM_acc} be given.
Then, replacing in \cref{defi:IDM_model} the acceleration of the follower in the following way 
\begin{align*}
    \Acc_{\text{ppa}}:\begin{cases}\mathcal{A}&\rightarrow\R\\
    (x,v,\xl,\vl)&\mapsto a-a\Big(\tfrac{|v|}{v_{0}}\Big)^{\delta}-a\Bigg(\tfrac{s_{0}+\max\left\{0,v\tau+\frac{v(\vl-v)}{2\sqrt{ab}}\right\}}{\xl-x-l}\Bigg)^{2}
    \end{cases}
\end{align*}
we call the resulting car following model the IDM with \textbf{p}artially \textbf{p}rojected \textbf{a}cceleration.
\end{definition}
However, as can be seen it does not prevent negative velocity in the case that \(\xl-x-l<s_{0}\) as the acceleration then becomes negative if the current follower's speed is zero, i.e., \(v=0\). Thus, we do not study it further.
\subsection{Velocity regularized acceleration}\label{subsec:velocity_regularized_acceleration}
Another improvement of the IDM is to add a regularization term which will make the third term in the acceleration function \cref{defi:IDM_acc} of the IDM in \cref{defi:IDM_model} become zero if the corresponding velocity approaches zero.
\begin{definition}[IDM with velocity regularized acceleration]\label{defi:IDM_velocity_regularized}
Given \cref{ass:input_datum}, we replace acceleration in \cref{defi:IDM_acc} in \cref{defi:IDM_model} by
\begin{align*}
    \Acc_{\text{vra}}:\begin{cases}\mathcal{A}&\rightarrow\R\\
    (x,v,\xl,\vl)&\mapsto a\bigg(1-\big(\tfrac{|v|}{v_{\text{free}}}\big)^{\delta}-h(v)\Big(\tfrac{2\sqrt{ab}(s_{0}+v\tau)+v(v-\vl)}{2\sqrt{ab}(\xl-x-l)}\Big)^{2}\bigg)
    \end{cases}
\end{align*}
with \(\mathcal{A}\) as in \cref{defi:IDM_acc} and regularization \(h\in W^{1,\infty}(\R;\R_{\geq0})\) be a monotonically increasing function satisfying \(h(0)=0\) and -- for a given \(\epsilon\in\R_{>0}\) -- \(h(v)=1\ \forall v\in\R_{\geq\epsilon}\). We call this the \textbf{velocity regularized acceleration IDM}.
\end{definition}
\begin{theorem}[Existence and Uniqueness of solutions on arbitrary time horizon ]\label{theo:existence_uniqueness_velocity_regularized}
Given \cref{ass:input_datum} and in addition assume that
\[
\exists v_{\min}\in\R_{>0}:\ \vl(t)\geq v_{\min}\ \forall t\in[0,T]\ \wedge\ h(v_{\min})>0
\]
the velocity regularized IDM in \cref{defi:IDM_velocity_regularized} admits on every finite time horizon \(T\in\R_{\geq0}\) a unique solution  satisfying \(x\in W^{2,\infty}((0,T))\)
\[
\xl(t)-x(t)-l\geq \min\Big\{\xlz-x_{0}-l,\sqrt{\tfrac{as_{0}^{2}h(v_{\min})}{-B}},\tfrac{-A+\sqrt{A^{2}+4aBh(v_{\min})s_{0}^{2}}}{2B}\Big\}\qquad\forall t\in[0,T]
\]
with
\begin{equation}
\begin{split}
    A&\: \tfrac{1}{2}(\vlz-v_{0})^{2}- B(\xlz-x_{0}-l)+a\tfrac{h(v_{\min})s_{0}^{2}}{\xlz-x_{0}-l}\\
    B&\: \essinf_{t\in[0,T]}\ul(t)-a
    \end{split}
    \label{defi:A_B_velocity_regularized}
\end{equation}
and additionally 
\[
 \max\{v_{\text{free}}, v_0\}\geqq \dot{x}\equiv v\geqq 0\ \text{ on } [0,T].
\]
If \(v_0 > 0\), it even holds 
\[v(t) >0 \quad \forall t \in [0, T].\]
\end{theorem}
\begin{proof}
To show the well-posedness it suffices to make sure that
\begin{itemize}[leftmargin=15pt]
\item \(\exists C\in\R_{>0}: \xlz - x_0 -l >0 \text{ implies that } \xl(t)-x(t)-l\geq C\textcolor{blue}{,}\ \forall t\in[0,T]\). 
Take the difference of the change of the leader's acceleration and the follower's to obtain and mimic somewhat the proof of \cref{theo:classical_IDM_safety_distance}
\begin{align}
    \ddot{\xl}(t)-\ddot{x}(t)&\geq \ul(t)-a+ah(v(t))\Big(\tfrac{2\sqrt{ab}(s_{0}+v(t)\tau)+v(t)(v(t)-\vl(t))}{2\sqrt{ab}\big(\xl(t)-x(t)-l\big)}\Big)^{2}\notag
    \intertext{assuming without loss of generality \(\vl(t)-v(t)\leq 0\), i.e., \(v(t)\geq \vl(t)\) on a certain time interval and as we have \(\inf_{t}\vl(t)\geq v_{\min}>0\), \(h\geqq 0\)}
    &\geq \essinf_{s\in[0,t]}\ul(s)-a+as_{0}^{2}h\big(v_{\min}\big)\Big(\tfrac{1}{\xl(t)-x(t)-l}\Big)^{2}\notag
    \end{align}
    and multiplying with \(\vl(t)-v(t)=\dot{\xl}(t)-\dot{x}(t)\leq 0\)
\begin{align}
    \Big(\ddot{\xl}(t)-\ddot{x}(t)\Big)(\dot{\xl}(t)-\dot{x}(t))&\leq \essinf_{s\in[0,t]}\ul(s)(\dot{\xl}(t)-\dot{x}(t))\label{eq:velocity_regularized_estimates_1}\\
    &\quad-a(\dot{\xl}(t)-\dot{x}(t))+ah\big(v_{\min}\big)s_{0}^{2}\tfrac{\dot{\xl}(t)-\dot{x}(t)}{\big(\xl(t)-x(t)-l\big)^{2}}.\label{eq:velocity_regularized_estimates_2}
\end{align}
Integrating over \((0,t)\) leads to 
\begin{align*}
    \tfrac{1}{2}(\dot{\xl}(t)-\dot{x}(t))^{2}&\leq \Big(\essinf_{s\in[0,t]} \ul(s)-a\Big)\Big(\big(\xl(t)-x(t)-l\big)-\big(\xlz-x_{0}-l\big)\Big)\\
    &\qquad+\tfrac{1}{2}(\vlz-v_{0})^{2} -ah\big(v_{\min}\big)s_{0}^{2}\Big(\tfrac{1}{\xl(t)-x(t)-l}-\tfrac{1}{\xlz-x_{0}-l}\Big)
    \intertext{or in shorter notation \(g\equivd\xl-x-l\)}
    \tfrac{1}{2}\big(g'(t)\big)^{2}&\leq A + Bg(t)-a\tfrac{h(v_{\min})s_{0}^{2}}{g(t)}
\end{align*}
    with \(A,B\) as in \cref{defi:A_B_velocity_regularized}.
Following the identical steps as in the proof of \cref{theo:classical_IDM_safety_distance}, particularly noticing that the left hand side is always nonnegative while the right hand side would go to \(-\infty\) if \(g\rightarrow 0\) we obtain the claimed lower bound on the distance.
\item \(\dot{x}(t)\geq 0\ \forall t\in[0,T]\). We know by the Picard-Lindel\"of theorem (\cite[Chapter 4]{Arnold} or \cite[Thm 1.3]{Coddington}) that there exists a solution on a significantly small time horizon \([0,T^{*}]\) with \(T^{*}\in\R_{>0}\). Assume that the velocity of a given vehicle \(\dot{x}\) becomes zero at a given time \(T^{**}\in\R_{>0}\). Then, by continuity of the velocity and due to the nonnegativity of the initial velocity there exists a first time \(t\in[0,T^{**})\) so that \(v(t)=\dot{x}(t)=0\). Plugging this into the corresponding acceleration we obtain at that time
\begin{align*}
    \ddot{x}(t)=\dot{v}(t)&=\Acc_{\text{vra}}\big(x(t),\dot{x}(t),\xl(t),\dot{\xl}(t)\big)=\Acc_{\text{vra}}\big(x(t),0,\xl(t),\dot{\xl}(t)\big)=a>0,
\end{align*}
thanks to the assumption on \(h\) in \cref{defi:IDM_velocity_regularized}, namely \(h(0)=0\). This means that whenever the velocity approaches zero, the derivative is strictly positive so that the velocity can never become zero.
Assume that for some \(t\in[0,T]\) we have \(v(t) \geq v_{\text{free}}.\) Then, by the velocity regularized acceleration in \cref{defi:IDM_velocity_regularized}, we have \(\dot{v}(t)<0\) and can conclude
\[v(t) \leq \max\{v_{\text{free}}, v_0\} \quad \forall t \in [0, T].\]
This gives the claim.
\end{itemize}
\end{proof}

\begin{remark}[Proper choice of the regularization \(h\)]
A proper choice for \(h\) consists for \(\epsilon\in\R_{>0}\) of
\begin{equation}
h_{\eps}:\begin{cases}
    \R&\rightarrow [0,1]\\
    v&\mapsto \tfrac{v}{\eps}\mathds{1}_{\R_{\geq0}}(v)\cdot \mathds{1}_{\R_{\leq\eps}}(v) + \mathds{1}_{\R_{>\eps}}(v).
    \end{cases}\label{eq:h}
\end{equation}
This is illustrated in the following \cref{fig:h}. As can be seen this ``saturation'' function is only continuous, resulting in an acceleration function which is not differentiable. Obviously, this could be changed by smoothing \(h_{\eps}\). We do not go into details.
\begin{figure}
    \centering
    \begin{tikzpicture}[scale=0.6]
        \begin{axis}[
            width=0.5\textwidth,
            xlabel={\(x\)},
            ylabel={\(h_{\epsilon}(x)\)},
                 legend pos={north west}       ]
        \addplot[color=blue,line width=1.5pt,domain=-1:0,dashed] {0};
        \addplot[color=blue,line width=1.5pt,domain=0:1,dashed,forget plot] {x};
        \addplot[color=blue,line width=1.5pt,domain=1:1.5,dashed,forget plot] {1};
        \addplot[color=red,line width=1.5pt,domain=-1:0,dotted] {0};
        \addplot[color=red,line width=1.5pt,domain=0:0.5,dotted,forget plot] {2*x};
        \addplot[color=red,line width=1.5pt,domain=0.5:1.5,dotted,forget plot] {1};
        \addplot[color=green,line width=1.5pt,domain=-1:0,dash dot] {0};
        \addplot[color=green,line width=1.5pt,domain=0:0.1,dash dot, forget plot] {10*x};
        \addplot[color=green,line width=1.5pt,domain=0.1:1.5,dash dot,forget plot] {1};
        \addlegendentry{\(\epsilon=1\)};
        \addlegendentry{\(\epsilon=0.5\)};
        \addlegendentry{\(\epsilon=0.1\)};
        \end{axis}
    \end{tikzpicture}
    \caption{Different choices of the ``saturation'' function \(h_{\epsilon}\) for \(\eps\in\{1,0.5,0.1\}\) as suggested in \cref{eq:h}. }
    \label{fig:h}
\end{figure}
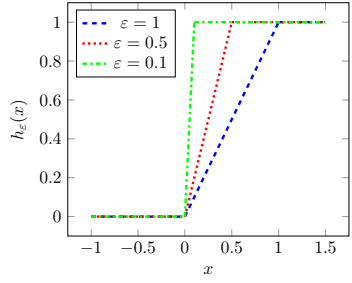
\end{remark}
We illustrate the model in the following
\begin{example}[Velocity regularized acceleration]
As can be seen in \crefrange{fig:velocity_regularized_acceleration_0}{fig:velocity_regularized_acceleration_2}, with this fix the spacing between the two vehicles is not getting large as had been observed in \crefrange{fig:projected_IDM_large_spacing_0}{fig:projected_IDM_large_spacing_2} but the follower speeds up immediately when there is enough safe distance to do so. The clipping due to the function \(h_{\epsilon}\) can be observed in particular in the acceleration profile which is nonsmooth. In all \crefrange{fig:velocity_regularized_acceleration_0}{fig:velocity_regularized_acceleration_2} we choose $\epsilon = 0.1$.
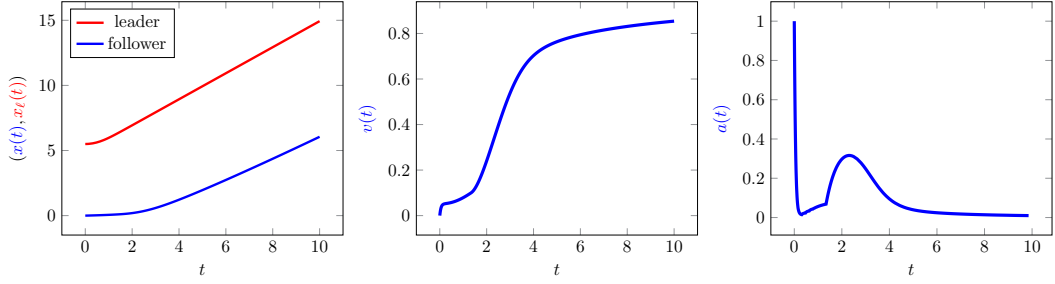
\begin{figure}
\centering
    \begin{tikzpicture}[scale=0.6]
        \begin{axis}[
            width=0.5\textwidth,
            xlabel={\(t\)},
            ylabel={\((\mathcolor{blue}{x(t)},\mathcolor{red}{\xl(t)})\)},
                 legend pos={north west}       ]
        \addplot[color=red,line width=1.5pt, mark=none]  table [y index=1, x index=0, col sep=comma] {Tikz_data/IDM_m2_positions_negative_v.csv};
        \addplot[color=blue,line width=1.5pt, mark=none]  table [y index=2, x index=0, col sep=comma] {Tikz_data/IDM_m2_positions_negative_v.csv};
        \addlegendentry{leader};
        \addlegendentry{follower};
        \end{axis}
    \end{tikzpicture}
      \begin{tikzpicture}[scale=0.6]
          \begin{axis}[
            width=0.5\textwidth,
            xlabel={\(t\)},
            ylabel={\(\mathcolor{blue}{v(t)}\)},
                 legend pos={north west}       ]
          \addplot[color=blue,line width=2pt, mark=none]  table [y index=1, x index=0, col sep=comma] {Tikz_data/IDM_m2_velocity_negative_v.csv};
          \end{axis}
          \end{tikzpicture}
     \begin{tikzpicture}[scale=0.6]
          \begin{axis}[
            width=0.5\textwidth,
            xlabel={\(t\)},
            ylabel={\(\mathcolor{blue}{a(t)}\)},
                 legend pos={north west}       ]
          \addplot[color=blue,line width=2pt, mark=none]  table [y index=1, x index=0, col sep=comma] {Tikz_data/IDM_m2_acc_negative_v.csv};
          \end{axis}
          \end{tikzpicture}
        \caption{The IDM with velocity regularized acceleration as in \cref{defi:IDM_velocity_regularized} and parameters as in \cref{fig:example_negative_speed_1}. Initial data are \(x_{0}=0,\ \xlz=l+1.5<l+s_{0},\ v_{0}=\vlz=0.\) 
        \textbf{Left:} the vehicles' position, \textbf{middle:} the vehicles' velicities, and \textbf{right:} the vehicles acceleration. 
        The \textcolor{red}{leader} follows the free flow acceleration.
        The \textcolor{blue}{follower} stays still and waits until there is a safe space to speed up.}
        \label{fig:velocity_regularized_acceleration_0}
\end{figure}
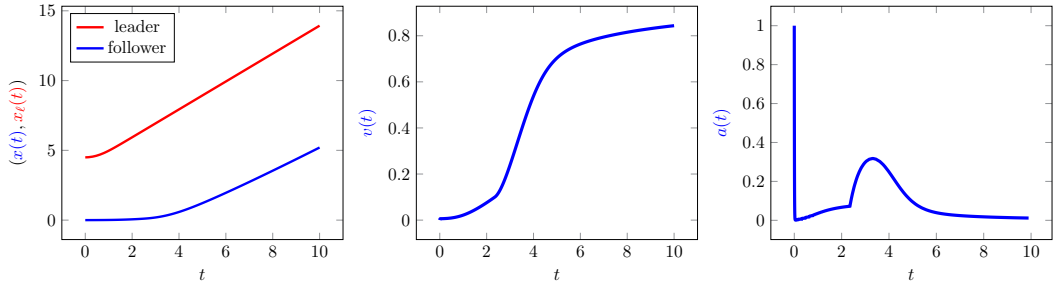
\begin{figure}
\centering
    \begin{tikzpicture}[scale=0.6]
        \begin{axis}[
            width=0.5\textwidth,
            xlabel={\(t\)},
            ylabel={\((\mathcolor{blue}{x(t)},\mathcolor{red}{\xl(t)})\)},
                 legend pos={north west}       ]
        \addplot[color=red,line width=1.5pt, mark=none]  table [y index=1, x index=0, col sep=comma] {Tikz_data/IDM_m2_positions_inf_v.csv};
        \addplot[color=blue,line width=1.5pt, mark=none]  table [y index=2, x index=0, col sep=comma] {Tikz_data/IDM_m2_positions_inf_v.csv};
        \addlegendentry{leader};
        \addlegendentry{follower};
        \end{axis}
    \end{tikzpicture}
      \begin{tikzpicture}[scale=0.6]
          \begin{axis}[
            width=0.5\textwidth,
            xlabel={\(t\)},
            ylabel={\(\mathcolor{blue}{v(t)}\)},
                 legend pos={north west}       ]
          \addplot[color=blue,line width=2pt, mark=none]  table [y index=1, x index=0, col sep=comma] {Tikz_data/IDM_m2_velocity_inf_v.csv};
          \end{axis}
          \end{tikzpicture}
     \begin{tikzpicture}[scale=0.6]
          \begin{axis}[
            width=0.5\textwidth,
            xlabel={\(t\)},
            ylabel={\(\mathcolor{blue}{a(t)}\)},
                 legend pos={north west}       ]
          \addplot[color=blue,line width=2pt, mark=none]  table [y index=1, x index=0, col sep=comma] {Tikz_data/IDM_m2_acc_inf_v.csv};
          \end{axis}
          \end{tikzpicture}
        \caption{The IDM with velocity regularized acceleration, with initial gap equal 0.5 and same parameters as in \cref{fig:example_negative_speed_1}. Both vehicles start with 0 velocity, and the leader follows the free flow IDM dynamics. \textbf{Left:} \(x_{0}=0,\ \xlz=l+0.5<l+s_{0},\ v_{0}=0=\vlz=0\). The follower stays still and waits until there is a safe space to speed up.
    \textbf{Right:} the follower remains still in the initial phase.}
    \label{fig:velocity_regularized_acceleration_1}
\end{figure}

\begin{figure}
\centering
    \begin{tikzpicture}[scale=0.6]
        \begin{axis}[
            width=0.5\textwidth,
            xlabel={\(t\)},
            ylabel={\((\mathcolor{blue}{x(t)},\mathcolor{red}{\xl(t)})\)},
                 legend pos={north west}       ]
        \addplot[color=red,line width=1.5pt, mark=none]  table [y index=1, x index=0, col sep=comma] {Tikz_data/IDM_m2_positions.csv};
        \addplot[color=blue,line width=1.5pt, mark=none]  table [y index=2, x index=0, col sep=comma] {Tikz_data/IDM_m2_positions.csv};
        \addlegendentry{leader};
        \addlegendentry{follower};
        \end{axis}
    \end{tikzpicture}
      \begin{tikzpicture}[scale=0.6]
          \begin{axis}[
            width=0.5\textwidth,
            xlabel={\(t\)},
            ylabel={\(\mathcolor{blue}{v(t)}\)},
                 legend pos={north west}       ]
          \addplot[color=blue,line width=2pt, mark=none]  table [y index=1, x index=0, col sep=comma] {Tikz_data/IDM_m2_velocity.csv};
          \end{axis}
          \end{tikzpicture}
     \begin{tikzpicture}[scale=0.6]
          \begin{axis}[
            width=0.5\textwidth,
            xlabel={\(t\)},
            ylabel={\(\mathcolor{blue}{a(t)}\)},
                 legend pos={north west}       ]
          \addplot[color=blue,line width=2pt, mark=none]  table [y index=1, x index=0, col sep=comma] {Tikz_data/IDM_m2_acc.csv};
          \end{axis}
          \end{tikzpicture}
        \caption{The IDM with velocity regularized acceleration, with initial gap equal 1. Both vehicles start with 0 velocity, and the leader follows the acceleration profile in figure \cref{fig:acc_leader}. \textbf{Left:} \(x_{0}=0,\ \xlz=l+\textcolor{red}{1}<l+s_{0},\ v_{0}=0=\vlz=0\).
    \textbf{Right:} again, the velocity always stays nonnegative.}
    \label{fig:velocity_regularized_acceleration_2}
\end{figure}
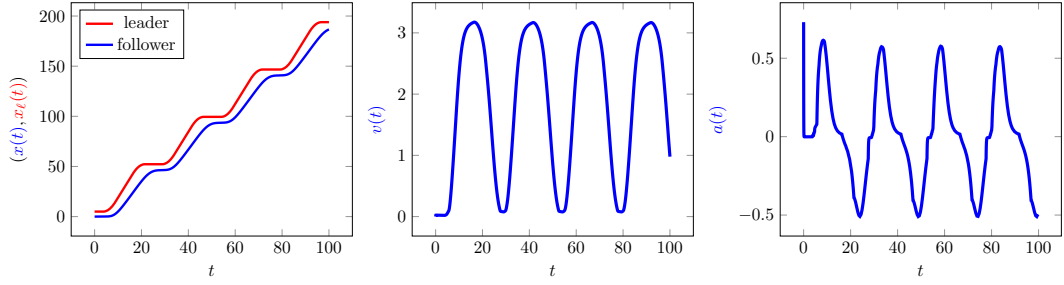
\end{example}
\subsection{A discontinuous improvement to prevent negative velocity}
Our last potential improvement for the IDM which had also been suggested a lot (see for instance \cite{IDM}) is an improvement which will become active only if the velocity becomes zero and the corresponding acceleration at the time where the velocity is zero is negative. In this way, the following improvement is the most natural one. We state it in the following \cref{defi:IDM_discontinuous}.
\begin{definition}[IDM with discontinuous acceleration]\label{defi:IDM_discontinuous}
Let \cref{ass:input_datum} and \(\Acc\) as in \cref{defi:IDM_acc} be given.
Then, replacing in \cref{defi:IDM_model} the acceleration of the follower in the following discontinuous way as
\begin{equation}
\begin{split}
    \dot{v}(t)&=\begin{cases}
                        \Acc(x(t),v(t),\xl(t),\vl(t))   & \text{if } v(t)> 0 \\
                        \Acc(x(t),v(t),\xl(t),\vl(t))     &\text{if }v(t)=0\ \wedge\ \xl(t) -x(t)-l\geq s_{0}\\
                        0 & \text{if }v(t)= 0 \ \wedge\ \xl(t)     -x(t)-l< s_{0}
                \end{cases},
\end{split}
\label{eq:IDM_discontinuous}
\end{equation}
we call this the \textbf{discontinuous IDM}.
\end{definition}
Although one might expect that the introduction of the discontinuity in \cref{eq:IDM_discontinuous} might prohibit a solution to exist for all times or also might destroy uniqueness, it actually does not as the following \cref{theo:IDM_discontinuous} states:
\begin{theorem}[Well-posedness of the discontinuously fixed IDM in \cref{defi:IDM_discontinuous}]\label{theo:IDM_discontinuous}
Given \cref{ass:input_datum} and assuming that the velocity of the leader is only zero at finitely many intervals, the discontinuous improvement of the IDM as in \cref{defi:IDM_discontinuous} admits a unique solution on every finite time horizon \(T\in\R_{>0}\) and satisfies
\[
x\in W^{2,\infty}((0,T)):\ \dot{x}\geqq0.
\]
\end{theorem}
In addition, the lower bound on the distance between follower and leader as in \cref{theo:classical_IDM_safety_distance} holds.
\begin{proof}

We consider several different cases:
\begin{itemize}[leftmargin=15pt]
\item Assume that at \(t=0\) we have \(v(0)=c > 0\). Obviously, for small time the solution of the system is unique as the right hand side is locally Lipschitz-continuous then. Either, \(v(t)>0\ \forall t\in[0,T]\). Then, there is nothing more to do as we never run into the discontinuity or \(\exists t^{*}\in(0,T]: v(t^{*})=0\). We distinguish two cases:
\begin{description}
    \item[\(\xl(t^{*})-x(t^{*})-l\geq s_{0}\):] However, in this case the right hand side has not changed so that the solution still exists and is unique. As at that time, the velocity is zero and as the leading car never moves backwards (by the assumption on the leader's velocity/acceleration, we can never go into the third case where it would hold \(\xl(t)-x(t)-l<s_{0}\) but only back into the first case with strictly positive velocity. In this case, the solution exists and is unique.
    \item[\(\xl(t^{*})-x(t^{*})-l< s_{0}\):] Then, the velocity of the follower is zero and we can only switch into the second case when the leader's position \(\xl\) increases, we are automatically left with the case that either we stay in the third case or that we move back to the second case. The second case, however, was already treated previously.
\end{description}
\item Assume that we have \(v(0)=0\). Then, we are either in item one or two of the previous case and are done.
\end{itemize}
As all of these changes only depend on the leader's trajectory \(\xl\) and velocity \(\vl\) which is given independent on the status of the follower, and we can conclude the existence and uniqueness of a solution.

The lower bound on the distance follows by the argument that the distance between follower and leader can only decrease if the followers velocity is not zero. However, then, we are in the ``classical'' IDM model and can apply exactly the reasoning of the proof of \cref{theo:classical_IDM_safety_distance}.
\end{proof}
\begin{example}[The discontinuous improvement]
In \crefrange{fig:discontinuous_0}{fig:discontinuous_2} we illustrate the dynamics of leader and follower for the discontinuous improvement proposed in \cref{theo:IDM_discontinuous}. Here, the velocity can become zero and stay zero for an amount of time dependent on the leaders velocity/position, but remains nonnegative. If the spacing between the two cars is large enough, the follower immediately speeds up and does not let the gap increase too much.

\begin{figure}
\centering
    \begin{tikzpicture}[scale=0.6]
        \begin{axis}[
            width=0.5\textwidth,
            xlabel={\(t\)},
            ylabel={\((\mathcolor{blue}{x(t)},\mathcolor{red}{\xl(t)})\)},
                 legend pos={north west}       ]
        \addplot[color=red,line width=1.5pt, mark=none]  table [y index=1, x index=0, col sep=comma] {Tikz_data/IDM_m4_positions_negative_v.csv};
        \addplot[color=blue,line width=1.5pt, mark=none]  table [y index=2, x index=0, col sep=comma] {Tikz_data/IDM_m4_positions_negative_v.csv};
        \addlegendentry{leader};
        \addlegendentry{follower};
        \end{axis}
    \end{tikzpicture}
      \begin{tikzpicture}[scale=0.6]
          \begin{axis}[
            width=0.5\textwidth,
            xlabel={\(t\)},
            ylabel={\(\mathcolor{blue}{v(t)}\)},
                 legend pos={north west}       ]
          \addplot[color=blue,line width=2pt, mark=none]  table [y index=1, x index=0, col sep=comma] {Tikz_data/IDM_m4_velocity_negative_v.csv};
          \end{axis}
          \end{tikzpicture}
     \begin{tikzpicture}[scale=0.6]
          \begin{axis}[
            width=0.5\textwidth,
            xlabel={\(t\)},
            ylabel={\(\mathcolor{blue}{a(t)}\)},
                 legend pos={north west}       ]
          \addplot[color=blue,line width=2pt, mark=none]  table [y index=1, x index=0, col sep=comma] {Tikz_data/IDM_m4_acc_negative_v.csv};
          \end{axis}
          \end{tikzpicture}
        \caption{The discontinuous improvement of the IDM, with initial gap equal 1.5 and same parameters as in \cref{fig:example_negative_speed_1}. Both vehicles start with 0 velocity, and the leader follows the acceleration profile in figure \cref{fig:acc_leader}. \textbf{Left:} \(x_{0}=0,\ \xlz=l+1.5<l+s_{0},\ v_{0}=0=\vlz=0\).
    \textbf{Right:} again, the velocity always stays nonnegative.}
    \label{fig:discontinuous_0}
\end{figure}
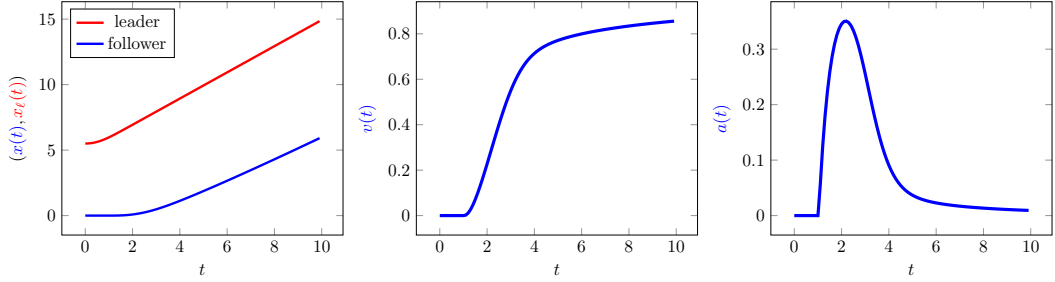

\begin{figure}
\centering
    \begin{tikzpicture}[scale=0.6]
        \begin{axis}[
            width=0.5\textwidth,
            xlabel={\(t\)},
            ylabel={\((\mathcolor{blue}{x(t)},\mathcolor{red}{\xl(t)})\)},
                 legend pos={north west}       ]
        \addplot[color=red,line width=1.5pt, mark=none]  table [y index=1, x index=0, col sep=comma] {Tikz_data/IDM_m4_positions_inf_v.csv};
        \addplot[color=blue,line width=1.5pt, mark=none]  table [y index=2, x index=0, col sep=comma] {Tikz_data/IDM_m4_positions_inf_v.csv};
        \addlegendentry{leader};
        \addlegendentry{follower};
        \end{axis}
    \end{tikzpicture}
      \begin{tikzpicture}[scale=0.6]
          \begin{axis}[
            width=0.5\textwidth,
            xlabel={\(t\)},
            ylabel={\(\mathcolor{blue}{v(t)}\)},
                 legend pos={north west}       ]
          \addplot[color=blue,line width=2pt, mark=none]  table [y index=1, x index=0, col sep=comma] {Tikz_data/IDM_m4_velocity_inf_v.csv};
          \end{axis}
          \end{tikzpicture}
     \begin{tikzpicture}[scale=0.6]
          \begin{axis}[
            width=0.5\textwidth,
            xlabel={\(t\)},
            ylabel={\(\mathcolor{blue}{a(t)}\)},
                 legend pos={north west}       ]
          \addplot[color=blue,line width=2pt, mark=none]  table [y index=1, x index=0, col sep=comma] {Tikz_data/IDM_m4_acc_inf_v.csv};
          \end{axis}
          \end{tikzpicture}
        \caption{The discontinuous improvement of the IDM, with initial gap equal 0.5 and same parameters as in \cref{fig:example_negative_speed_1}. Both vehicles start with 0 velocity, and the leader follows the acceleration profile in figure \cref{fig:acc_leader}. \textbf{Left:} \(x_{0}=0,\ \xlz=l+0.5<l+s_{0},\ v_{0}=\vlz=0\).
    \textbf{Right:} again, the velocity always stays nonnegative.}
    \label{fig:discontinuous_1}
\end{figure}
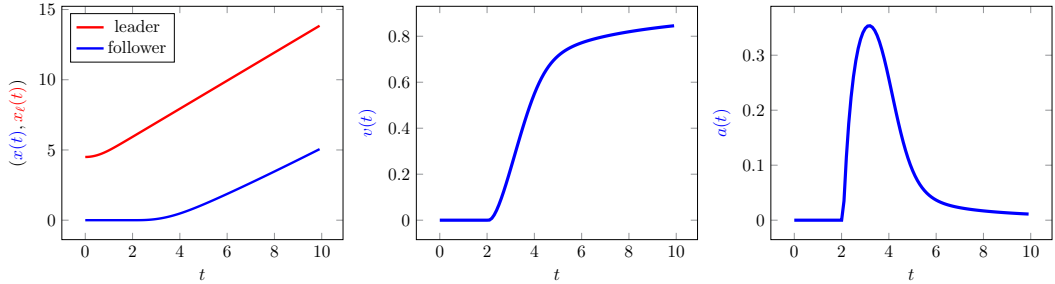

\begin{figure}
\centering
    \begin{tikzpicture}[scale=0.6]
        \begin{axis}[
            width=0.5\textwidth,
            xlabel={\(t\)},
            ylabel={\((\mathcolor{blue}{x(t)},\mathcolor{red}{\xl(t)})\)},
                 legend pos={north west}       ]
        \addplot[color=red,line width=1.5pt, mark=none]  table [y index=1, x index=0, col sep=comma] {Tikz_data/IDM_m4_positions.csv};
        \addplot[color=blue,line width=1.5pt, mark=none]  table [y index=2, x index=0, col sep=comma] {Tikz_data/IDM_m4_positions.csv};
        \addlegendentry{leader};
        \addlegendentry{follower};
        \end{axis}
    \end{tikzpicture}
      \begin{tikzpicture}[scale=0.6]
          \begin{axis}[
            width=0.5\textwidth,
            xlabel={\(t\)},
            ylabel={\(\mathcolor{blue}{v(t)}\)},
                 legend pos={north west}       ]
          \addplot[color=blue,line width=2pt, mark=none]  table [y index=1, x index=0, col sep=comma] {Tikz_data/IDM_m4_velocity.csv};
          \end{axis}
          \end{tikzpicture}
     \begin{tikzpicture}[scale=0.6]
          \begin{axis}[
            width=0.5\textwidth,
            xlabel={\(t\)},
            ylabel={\(\mathcolor{blue}{a(t)}\)},
                 legend pos={north west}       ]
          \addplot[color=blue,line width=2pt, mark=none]  table [y index=1, x index=0, col sep=comma] {Tikz_data/IDM_m4_acc.csv};
          \end{axis}
          \end{tikzpicture}
        \caption{The discontinuous improvement of the IDM, with initial gap equal 1. Both vehicles start with 0 velocity, and the leader follows the acceleration profile in figure \cref{fig:acc_leader}. \textbf{Left:} \(x_{0}=0,\ \xlz=l+\textcolor{red}{1}<l+s_{0},\ v_{0}=\vlz=0\).
    \textbf{Middle:} again, the velocity always stays nonnegative. \textbf{Right:} acceleration}
    \label{fig:discontinuous_2}
\end{figure}
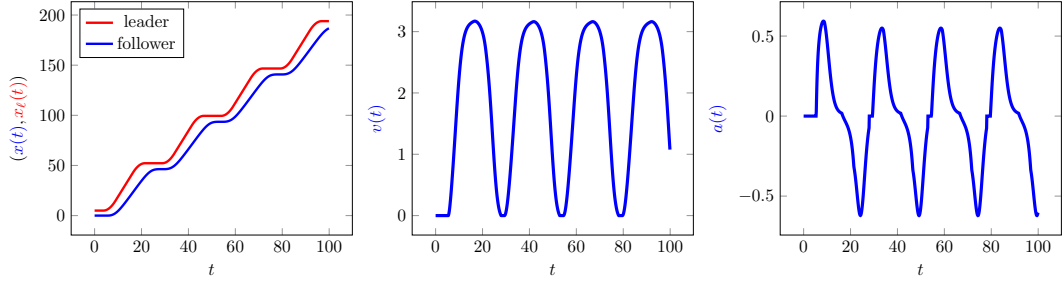

\end{example}

\subsection{Comparisons among modifications}

We herein summarize the strength and weakness of each modification.
\begin{itemize}
\item The acceleration projected IDM in \cref{defi:IDM_projected_velocities_acceleration} is the most straightforward modification to fix the negative velocities. However as shown in \cref{fig:projected_IDM_large_spacing_0,fig:projected_IDM_large_spacing}, it suffers from the fact that the follower waits too long to start driving. It may also lead to physical unreasonability issue with certain initial datum as we point out in \cref{exa:acceleration_projected_IDM_unreasonable}.

\item The velocity regularized IDM in \cref{defi:IDM_velocity_regularized} is capable of preventing negative velocity, and the follower speeds up immediately when it observes safe distance. However in this modification one needs to specify the regularization functions $h$, which might require extra parameter tuning.

\item The discontinuous improvement of the IDM as in \cref{defi:IDM_discontinuous} prohibits negative velocities, without the need to introduce extra saturation functions. 
\end{itemize}

In order to quantitatively compare the modifications, we compute the average distance (as well as the variance) between the leader and follower over time in \cref{tab:comparisons_among_modifications} on three different scenarios we tested. As it turns out the differences between the latter two modifications are minimal in all scenarios and both mean/variance. On the other hand, the acceleration projected IDM has greater average distance and variance, verifying the potential issue that with the modification, the following vehicle waits for too long to start driving, and reacts to the leader's velocity change in a slower manner.

\begin{table}
\centering
	\caption{The average distance between the leader and follower (and the variance reported in the parenthesis)}
    	\begin{tabular}{l||r||r|
    	| r}
		\textbf{Average distance (variance)} & \cref{item:case_1}& \cref{item:case_2}& \cref{item:case_3}\\
		\hline\hline
		acceleration projected IDM \cref{defi:IDM_projected_velocities_acceleration}& 7.99 (1.19)   & 8.09 (3.55)    &14.94 (54.99)\\
		velocity regularized IDM \cref{defi:IDM_velocity_regularized}& 7.76 (1.00)   &   7.24 (1.70) & 12.81 (25.51)\\
		discontinuous IDM \cref{defi:IDM_discontinuous}& 7.75 (1.01) &7.31 (1.75) & 12.39 (25.18)\\
    \end{tabular}
	\label{tab:comparisons_among_modifications}
\end{table}

\section{Generalization to many cars}\label{sec:multi_vehicle}
In the proposed framework, we have only studied the case where we have one leader and one follower and the leader (in most cases) satisfies an acceleration profile where their velocity is nonnegative. However, as this is somewhat arbitrary, all proposed results and all ``improvements'' remain valid as long as the velocity of the follower remains non-negative when we generalize this to more than two cars. We make this precise in the following \cref{defi:car_following_multi} but first introduce the number of cars as well as some physical reasonable assumption on the input datum:
\begin{assumption}[Input datum for multiple cars]\label{ass:input_datum_many_cars}
Let \(N\in\N_{\geq1}\) be given. We assume that the parameters of the \(\Acc\) satisfy what we have assumed in \cref{ass:input_datum} as well as the leaders acceleration \(u_{\text{lead}}\). Additionally, we assume for the initial datum (position and velocity)
\[
\bx_{0}\in\R^{N}: \bx_{0,i}-\bx_{0,i+1}> l\ \forall i\in\{1,\ldots,N-1\}\ \wedge\ \bv_{0}\in\R_{\geq0}^{N}.
\]
\end{assumption}

\begin{definition}[The car-following model for many cars]\label{defi:car_following_multi}
Given \cref{ass:input_datum_many_cars}, \cref{defi:IDM_acc} and \cref{ass:input_datum}. Then, for \(\ul\in\mathcal{U}_{\text{lead}}\) the dynamics for the IDM with many vehicles \(N\in\N_{\geq1}\) read in \((\bx,\bv):[0,T]\rightarrow\R^{N}\times\R^{N}\) as
\begin{equation}
\begin{aligned}
    \dot{\bx}_{1}(t)&=\bv_{1}(t),&&&& t\in[0,T]\\
    \dot{\bv}_{1}(t)&=\ul(t),&&&& t\in[0,T]\\
    \dot{\bx}_{i}(t)&=\bv_{i}(t),&& i\in\{2,\ldots,N\} && t\in[0,T]\\
    \dot{\bv}_{i}(t)&=\Acc(\bx_{i}(t),\bv_{i}(t),\bx_{i-1}(t),\bv_{i-1}(t)), && i\in\{2, \ldots, N\}&& t\in[0,T] \\
    \bx(0) &=\bx_{0},\\
    \bv(0) &= \bv_{0}.
\end{aligned}
\end{equation}
\end{definition}
The system for many cars is illustrated in \cref{fig:car_following_multi}.
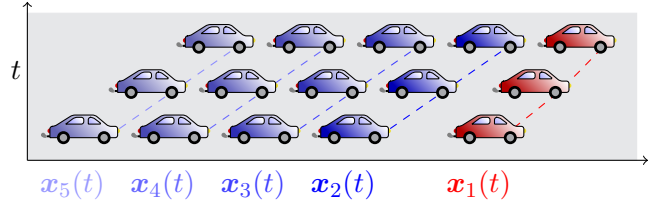
\begin{figure}
    \centering
\begin{tikzpicture}[scale=0.6]
 \draw[fill,color=gray!20!white](-6.5,-1)--(-6.5,2.25)--(7,2.25)--(7,-1)--cycle;
 \draw[->] (-6.5,-1) -- (-6.5,2.5);
\draw (-6.75,1) node {\(t\)};
\draw[->] (-6.5,-1) -- (7.25,-1);
\draw[dashed,blue] (1.25,-0.5) -- (4.5,1.75);
\draw[dashed,blue!75!white] (-.75,-0.5) -- (2.5,1.75);
\draw[dashed,blue!60!white] (-2.75,-0.5) -- (.5,1.75);
\draw[dashed,blue!45!white] (-4.75,-0.5) -- (-1.5,1.75);
\draw [dashed,red] plot [smooth] coordinates { (4,-0.5) (5.25,0.5) (6.25,1.5)};
\draw (3.5,-1) node[below] {\textcolor{red}{\(\bx_{1}(t)\)}};
\draw (0.5,-1) node[below] {\textcolor{blue}{\(\bx_{2}(t)\)}};
\draw (-1.5,-1) node[below] {\textcolor{blue!75!white}{\(\bx_{3}(t)\)}};
\draw (-3.5,-1) node[below] {\textcolor{blue!60!white}{\(\bx_{4}(t)\)}};
\draw (-5.5,-1) node[below] {\textcolor{blue!40!white}{\(\bx_{5}(t)\)}};
\begin{scope}[shift={(1.5,0)}]\car{blue}\end{scope}
  \begin{scope}[shift={(5,1)}]\car{red}\end{scope}
 \begin{scope}[shift={(3,1)}]\car{blue}\end{scope}
 \begin{scope}[shift={(4,0)}]\car{red}\end{scope}
 \begin{scope}[shift={(0,-1)}]\car{blue}\end{scope}
 \begin{scope}[shift={(3,-1)}]\car{red}\end{scope}
 \begin{scope}[shift={(-2,-1)}]\car{blue!75!white}\end{scope}
 \begin{scope}[shift={(-0.5,0)}]\car{blue!75!white}\end{scope}
 \begin{scope}[shift={(1,1)}]\car{blue!75!white}\end{scope}
  \begin{scope}[shift={(-4,-1)}]\car{blue!60!white}\end{scope}
 \begin{scope}[shift={(-2.5,0)}]\car{blue!60!white}\end{scope}
 \begin{scope}[shift={(-1,1)}]\car{blue!60!white}\end{scope}
 \begin{scope}[shift={(-6,-1)}]\car{blue!45!white}\end{scope}
 \begin{scope}[shift={(-4.5,0)}]\car{blue!45!white}\end{scope}
 \begin{scope}[shift={(-3,1)}]\car{blue!45!white}\end{scope}
 \end{tikzpicture}
\caption{The leader \textcolor{red}{\(\bx_{1}(t)\)} with its dynamics determined by the acceleration \(u_{\text{lead}}(t)\) and the following cars \textcolor{blue}{\(\bx_{2}(t)\)}, \textcolor{blue!75!white}{\(\bx_{3}(t)\)},  \textcolor{blue!60!white}{\(\bx_{4}(t)\)}, \textcolor{blue!45!white}{\(\bx_{5}(t)\)} with its dynamics governed by the classical IDM \cref{defi:IDM_model}. The overall dynamics is stated in \cref{defi:car_following_multi}.}
    \label{fig:car_following_multi}
\end{figure}
Having the definition, we obtain the following general result on the well-posedness when applying the proper acceleration functions -- IDM improvements -- introduced before:
\begin{theorem}[Well-posedness of some of the previously discussed models]
Let \cref{ass:input_datum_many_cars} hold and
consider as ``improvement'' of the IDM either the 
\begin{itemize}[leftmargin=15pt]
    \item the \textbf{velocity regularized acceleration} in \cref{defi:IDM_velocity_regularized} with the additional assumption that there exists \(v_{\text{min}} \in \R_{>0}\) such that \(\vl \geqq v_{\min}\) on \([0, T]\) and \(\bv_0 \in \R_{>0}^N\)
    \item the \textbf{discontinuous improvement} in \cref{defi:IDM_discontinuous}.
\end{itemize}
Then, the system of \(2N\) coupled initial value problems (\(N\in\N_{\geq1}\)) admits a unique solution on every given time horizon \(T\in\R_{>0}\) and the solution satisfies for the \textbf{velocity regularized acceleration}
\begin{align*}
    \bx\in W^{2,\infty}\big((0,T);\R^{N}\big):\ \dot{\bx}\equivd\bv>\boldsymbol{0} \text{ on } [0,T]
\end{align*}
and for the \textbf{discontinuous improvement}
\begin{align*}
    \bx\in W^{2,\infty}\big((0,T);\R^{N}\big):\ \dot{\bx}\equivd\bv\geqq\boldsymbol{0} \text{ on } [0,T].
\end{align*}
\end{theorem}
\begin{proof}
The proof consists of recalling that the proofs of the corresponding improvements all work for any leaders acceleration proposed in \cref{ass:input_datum_many_cars} as long as the corresponding velocity would remain nonnegative/strictly positive.
As the dynamics are only ``one-directionally'' coupled, i.e. \ the dynamics of the follower depend on the leader but not vice versa, we can use an
inductive argument and first look on the dynamics for \(\bx_{1},\bx_{2}\). According to \cref{theo:existence_uniqueness_velocity_regularized} and \cref{theo:IDM_discontinuous} we obtain the well-posedness. However, additionally also that the velocity \(\bv_{2}\) is nonnegative/strictly greater than zero. As a next step, we can thus consider the dynamics of \(\bx_{2},\bx_{3}\). As \(\bx_{2}\) satisfies the non-negativity/strict positivity of the velocity we can again apply the stated well-posedness results for two vehicles as in \cref{theo:existence_uniqueness_velocity_regularized} and \cref{theo:IDM_discontinuous}. This procedure can be iterated until we have reached \(\bx_{N}\) and we conclude with the existence and uniqueness of solutions and the nonnegativity/strict positivity of the velocities.
\end{proof}

\section{Conclusions and future work}\label{sec:conclusions}
In this contribution, we have demonstrated the ill-posedness of the rather often used \textbf{intelligent driver model} (IDM) for specific initial datum and have presented some improvements to avoid these problems.
The proposed work builds a solid foundation for future work on
\textbf{1)} Multi-lane traffic with lane-changing. The lane-changing requires to know under which condition a car can change lane and how the well-posedness is affected by the lane change. For optimizing lane-changing to smooth traffic and avoid stop and go waves (thus saving energy) we will next consider a suitable hybrid optimal control problem based on the proposed dynamics.
\textbf{2)} Signalized junctions/intersection modelled with a form of the IDM also requires the proposed well-posedness in particular as a red traffic light necessitates the vehicles in front of it to stop (velocity becomes zero).
\textbf{3)} Stability of solutions with regard to the model parameters and comparison of stability for the different suggested improvements.
\textbf{4)}
Implementations of the models presented in this article. While the numerical work presented above results from \texttt{matlab} \cite{MATLAB} implementations and the use of \texttt{ODE45} and similar routines, it would be interesting to also theoretically study the discretization of these equations with standard finite difference schemes to see what guarantees can be provided for the numerical solutions (for example order of the numerical schemes, error bounds on the numerical solution etc.). 
\textbf{5)} It would also be of great interest for the improvements of the IDM model presented here to validate them against field data, or to see if the resulting microsimulation implementations (SUMO \cite{SUMO2018}, Aimsun \cite{aimsun}, etc.) match experimental data for specific choices of numerical parameters. In particular, it would be interesting to measure discrepancies with the prior IDM improvements in their implementation, and compare their respective performances. 
\textbf{6)} Finally, obtaining a general well-posedness result for the original IDM when restricting specific parameters and input data would be important and will be subject to further study to justify many already available simulation results a posteriori.
\section*{Acknowledgements}
We would like to thank the reviewers for their suggestions for improvement. In particular the suggestions led to a rewriting of \cref{sec:IDM} and a well-posedness result for the classical IDM.

 This material is based upon work supported by the National Science Foundation under Grant Numbers CNS-1837244 (A.~Bayen), CNS-1837481 (B.~Piccoli).
 The research is based upon work supported by the U.S.~Department of Energy’s Office of Energy Efficiency and Renewable Energy (EERE) under the Vehicle Technologies Office award number CID DE-EE0008872.
 The views expressed herein do not necessarily represent the views of the U.S.~Department of Energy or the United States Government.
\bibliographystyle{plain}
\bibliography{IDM.bib}

\end{document}

%% file: definitions.tex
\newcommand{\R}{\mathbb R}
\newcommand{\N}{\mathbb N}

\newcommand{\ul}{u_{\text{lead}}}
\newcommand{\xl}{x_{\ell}}
\newcommand{\xlz}{x_{\ell,0}}
\newcommand{\vl}{v_{\ell}}
\newcommand{\vlz}{v_{\ell,0}}

\newcommand{\bv}{\boldsymbol{v}}

\newcommand{\eps}{\epsilon}

\renewcommand{\phi}{\varphi}

\newcommand{\bx}{\boldsymbol{x}}

\newcommand{\Acc}{\mathrm{Acc}}
\DeclareMathOperator{\Accfront}{\ensuremath{\mathrm{Acc_{front}}}}

\renewcommand{\:}{\mathrel{\coloneqq}}

\newcommand{\equivd}{\ensuremath{\vcentcolon\equiv}}

\newcommand{\dd}{\ensuremath{\,\mathrm{d}}}

\DeclareMathOperator{\sgn}{sgn}

\DeclareMathOperator*{\essinf}{ess-\inf}

\renewcommand{\epsilon}{\varepsilon}



%% file: define_cars.tex
\definecolor{golden}{RGB}{253,181,21}

\newcommand{\car}[1]{
    \draw[draw=black,fill=yellow!90!black,thick,color=yellow!90!black] (1.5,.65) circle (.03);
        \draw[draw=black,fill=red!80!black,thick,color=red!80!black] (
    0.0,.65) circle (.03);
    
    \draw[black]
        [draw=black,fill=black,rounded corners=0.0ex, thin](0,0.52)--(-0.05,0.52)--(-0.05,0.54)--(0.0,0.54);
        
        \draw[draw=black,fill=white!50!black,thick,color=white!50!black] (
    -0.15,.56) circle (.025);    
        \draw[draw=black,fill=white!50!black,thick,color=white!50!black] (
    -0.13,.54) circle (.015); 
        \draw[draw=black,fill=white!50!black,thick,color=white!50!black] (
    -0.10,.53) circle (.01);     
    
    \shade[top color=#1!70!black, bottom color=white, shading angle={135}]
        [draw=black,fill=red!20,rounded corners=0.2ex, thin](0,0.5)--(1.5,0.5)--(1.5,0.75)--(1.3,0.75)--(1,1.0)--(0.5,1)--(0.1,0.75)--(0.0,0.75)--(0,0.5);
    \shade[top color=blue!40!white, bottom color=white, shading angle={135}]
        [draw=black,fill=red!20,rounded corners=0.2ex, thin](0.35,0.8)--(0.75,0.8)--(0.75,0.95)--(0.55,0.95)--(0.35,0.8);
    \shade[top color=blue!30!white, bottom color=white, shading angle={135}]
        [draw=black,fill=red!20,rounded corners=0.2ex, thin](0.85,0.8)--(1.15,0.8)--(0.95,0.95)--(0.85,0.95)--(0.85,0.8);
    \draw[draw=black,fill=gray!50,thick] (0.3,.5) circle (.12);
    \draw[draw=black,fill=gray!50,thick] (1.2,.5) circle (.12);    
    \draw[draw=black,fill=gray!80,semithick] (0.3,.5) circle (.1);
    \draw[draw=black,fill=gray!80,semithick] (1.2,.5) circle (.1);
}